\documentclass{amsart}
\usepackage{mathrsfs}
\usepackage{graphicx}
\usepackage{amsmath}
\usepackage{amscd}
\usepackage{amssymb}
\usepackage[all]{xy}

\input xy
\xyoption{all}

\newcommand{\pd}{\partial}
\newcommand{\bC}{{\mathbb C}}

\newcommand{\bt}{{\mathbf t}}

\newcommand{\bZ}{{\mathbb Z}}

\newcommand{\cM}{{\mathcal M}}

\newcommand{\half}{\frac{1}{2}}

\newcommand{\Mbar}{\overline{\cM}}

\DeclareMathOperator{\Aut}{Aut}

 \DeclareMathOperator{\tr}{tr}
 
 \DeclareMathOperator{\val}{val}

\newtheorem{thm}{Theorem}[section]
\newtheorem{theorem/definition}{Theorem/Definition}[section]

\newtheorem{prop}{Proposition}[section]

\theoremstyle{remark}
\newtheorem{remark}{Remark}[section]

\theoremstyle{definition}

\newcommand{\be}{\begin{equation}}
\newcommand{\ee}{\end{equation}}
\newcommand{\bea}{\begin{eqnarray}}
\newcommand{\ben}{\begin{eqnarray*}}
\newcommand{\een}{\end{eqnarray*}}
\newcommand{\eea}{\end{eqnarray}}
\newcommand{\bet}{\begin{equation}
\begin{split}}
\newcommand{\eet}{\end{split}
\end{equation}}

\usepackage{color}

\definecolor{yellow}{rgb}{1,1,0}
\definecolor{orange}{rgb}{1,.7,0}
\definecolor{red}{rgb}{1,0,0} \definecolor{green}{rgb}{0,1,1}
\definecolor{white}{rgb}{1,1,1}

\definecolor{A}{rgb}{.75,1,.75}

\newcommand{\corr}[1]{\langle {#1} \rangle}

\begin{document}

\title
{Fat and Thin Emergent Geometries  of Hermitian One-Matrix Models}

\author{Jian Zhou}
\address{Department of Mathematical Sciences\\
Tsinghua University\\Beijing, 100084, China}
\email{jzhou@math.tsinghua.edu.cn}

\begin{abstract}
We use genus zero free energy functions
of Hermitian matrix models to define spectral curves and their special deformations.
They are special plane curves defined by formal power
series with integral coefficients generalizing the Catalan numbers.
This is done in two different versions, depending on two different genus expansions,
and these two versions are in some sense dual to each other.
\end{abstract}
\maketitle

\section{Introduction}

Studies of Hermitian one-matrix models,
based on Gaussian type integrals on the spaces
of Hermitian $N\times N$-matrices,
are often  focused on their limiting behaviors
as $N \to \infty$ in the literature.
In the early days of string theory,
with the applications of fat graphs introduced by 't Hooft \cite{tH},
a more sophisticated limit called the double scaling limit
played an important role in relating matrix models
to a wide range of objects interesting to mathematicians and mathematical physicists,
including orthogonal polynomials, integrable hierarchies, Virasoro constraints,
etc.
See e.g. \cite{DGZ} for a survey.

Since its introduction by Euler,
graphs have been widely studied by mathematicians,
especially in combinatorics,
and their studies form an important part of discrete mathematics.
They have also been widely used in quantum field theories
to formulate Feynman rules for Feynman sums.
Of course graph theory are developed by mathematicians and physicists
for different purposes and by different techniques.
So it is reasonable to expect that techniques developed
in quantum field theories should be useful
to combinatorial problems related to the enumerations of graphs,
and vice versa.
See e.g. \cite{BIZ} for an exposition of quantum field techniques
in graph enumerations.

By contrast,
fat-graphs, even though they are also  of combinatorial and physical origins,
have  more direct geometrical connections.
They have also long been studied by combinatorists \cite{Tutte},
even before their appearance in the physics literature.
In the combinatorics literature they have been called maps, or
graphs on Riemann surfaces.
However the  significance of such geometric connections
seemed to be only realized after
the advent of string theory in connection with 2D quantum gravity.
They were used by Harer-Zagier \cite{Har-Zag} and Penner \cite{Penner}
to calculate the orbifold Euler characteristics of the moduli spaces
of Riemann surfaces $\cM_{g,n}$,
making a surprising connection to special values of Riemann zeta-function.
(Of course, another connection between random matrices and Riemann zeta-function
was made by Montgomery and Dyson even earlier \cite{Montgomery}.
See \cite{BKSWZ} for an exposition.)
Through the work of Witten \cite{Witten1, Witten2},
matrix models are related to intersection theory on the Deligne-Mumford
compactifications $\Mbar_{g,n}$ and its generalizations to include the spin structures,
and connections to the KdV hierarchies and Virasoro constraints,
first discovered in the setting of Hermitian matrix models,
was made by Witten for the geometric theories related to $\Mbar_{g,n}$
and its generalizations.
See the references in Witten \cite{Witten1, Witten2} for a guide to the literature on this
line of developments.
Kontsevich \cite{Kon} introduced a new kind of matrix models to prove Witten's Conjecture.

The generalizations of Witten Conjcture/Kontsevich Theorem to relate
other Gromov-Witten type geometric theories to integrable hierarchies
and to establish their infinite-dimensional algebraic constraints
have been one of the driving forces in the mathematical theory of string theory.
Another driving force is the study of mirror symmetry,
among which the local mirror symmetry of toric Calabi-Yau 3-folds
has been studied via matrix models.
The mirror geometry of a toric Calabi-Yau 3-fold
can be encoded in a plane curve.
Dijkgraaf and Vafa \cite{Dij-Vaf1, Dij-Vaf2}
identify this curve with the spectral curve of a suitable matrix model.

It is interesting to note that the notion of a spectral curve
also appear in at least two other settings.
One is in the setting of integrable hierarchies (see e.g. \cite{Mulase}),
the other is in the setting of Eynard-Orantin topological recursions
\cite{EO} which itself grows out of the theory of random matrices.
So far,
both Witten Conjecture/Kontsevich Theorem and local mirror symmetry
have been studied from the following three points of view:
matrix models, integrable hierarchies, and EO topological recursions.
Spectral curves have played an important role
in all these different approaches.

Borrowing terminology from statistical physics,
we say that the appearance of spectral curves in matrix models are emergent.
By emergent geometry we mean the geometric structures that appear
when one has an infinite  degree of freedom.
The spectral curve of matrix models reflects the collective behavior of the spectrum
of the $N\times N$-Hermitian matrix as $N$ goes to infinity.
In the literature,
one introduces a discrete resolvent $\omega_N$
for $N \times N$-Hermitian matrix and assume it tends to a differentiable function
$\omega(z)$ as $n \to \infty$.
This function satisfies a quadratic equation,
hence defining a hyperelliptic plane curves.
This is how spectral curves arise in matrix models.
See e.g. \cite[p. 47]{DiF}.

A natural question one can ask is whether the emergence of spectral curves  can only
occur when one takes the limit $N\to \infty$,
or it can also occur for finite $N$.
This is a legitimate question because for any theory with infinite degree of freedom
one can ask this type of questions,
and for each finite $N$,
the Hermitian one-matrix model based on $N \times N$-matrices
can be regarded as a formal quantum field theory
whose operator algebra is the algebra of symmetric functions
\cite[\S 2]{Zhou-Mat-Mod},
and so it has an infinite degree of freedom.

Another motivation for asking this question comes from our study
of the Witten Conjecture/Kontseivich Theorem from the point of view of
emergent geometry.
In an earlier work \cite{Zhou-WK} we have shown that
Virasoro constraints of \cite{DVV} satisfied by the Witten-Kontsevich $\tau$-function
is equivalent to EO toplogical recursion on the Airy curve.
This result was later understood from two different perspectives.
First in \cite{Zhou-Quan-Def},
this was interpreted as a geometric reformulation  of the Virasoro constraints
by introducing the notion of special deformations and quantization of
special deformations of the Airy curve.
Next, after the author became aware of the notion of emergence
through the contacts with some physicists working on
condense matter physics, he started to use the terminology of emergent
geometry in \cite{Zhou-Emergent}.
In that work,
the emergent geometry of Witten-Kontsevich partition function was cast
in a more general setting of emergent geometry of KP hierarchy
based on boson-fermion correspondence,
and in particular a formula for the bosonic $n$-point function
associated to  $\tau$-function of the KP hierarchy was derived.
See \cite{Zhou-KP-Emergent} for further developments.

We start to address this question in \cite{Zhou-1D}.
Instead of letting $N \to \infty$ as in the literature,
we go to the other extreme and let $N =1$.
The resulting theory is called the theory of topological 1D gravity
or the polymer model.
The partition function of this theory
is just a universal generating series that enumerates all possible graphs,
but we study it by quantum field theory techniques,
including Feynman sums, Feynman rules, Virasoro constraints
and KP hierarchy.
Furthermore,
we introduce an analogue of Wilson's theory of renormalization
in its study.
An interesting consequence is that we can reduce the problem of enumeration of
graphs with fixed number of loops to the problem
of determining finitely many correlators
in this theory.
Another result obtained in this work is that
even for $N=1$,
one can talk about the emergent geometry of spectral curves.
In this case,
we introduce a spectral curve which we call the signed Catalan curve
since it is related to Catalan numbers up to signs.
We also introduce its special deformation and its quantization.
In a more recent development,
we unify this theory with the theory of holomorphic anomaly equations
\cite{BCOV} via the diagrammatics of Deligne-Mumford compactification \cite{Wang-Zhou}.

Next, we move on the case of finite $N \geq 1$.
Since it is known that each of the partition function $Z_N$
of the Hermitian $N\times N$-matrix model is a $\tau$-function of
the KP hierarchy \cite{Shaw-Tu-Yen},
we can apply all the results on emergent geometry we have developed for KP hierarchy.
We can also try to apply the insight we gain from \cite{Zhou-1D}.
This is exactly what we do in \cite{Zhou-Mat-Mod}, \cite{Zhou-Fat-Thin} and this paper.
First of all,
we identify the element in the Sato grassmannian corresponding to $Z_N$
and obtain a formula for the associated bosonic $n$-point functions in \cite{Zhou-Mat-Mod}.
This was achieved by the following techniques well-known in the literature:
First, the correlators are expressed as an enumeration problem of fat graphs;
next, this is converted to a problem in 2D Yang-Mills theory with finite gauge groups
\cite{Dij-Wit},
and solved with the help of the representation theory of symmetric groups;
finally,
with the help of the theory of symmetric functions,
one can reformulate the result in terms of dimension formula for
irreducible representations of $U(N)$.
This version of Schur-Weyl duality goes in the reverse direction
to the duality discovered by Gross and Taylor \cite{Gro-Tay},
which goes from large $N$ 2D $U(N)$-Yang-Mills theory
to Hurwitz numbers of branched coverings which corresponds to 2D Yang-Mills theories
with symmetric groups as gauge groups.

Recall we are focusing on matrix models with finite $N$.
Once we identify the element in the Sato grassmannian that corresponds to
$Z_N$,
we can treat $N$ as a parameter that can take any real or complex values
to get a family of tau-functions of the KP hierarchy.
This point of view is particularly powerful when we combine
it with the introduction of 't Hooft coupling constant
which we do in the end of \cite{Zhou-Mat-Mod}.
This leads to a different genus expansion in Hermitian matrix models.
In \cite{Zhou-Fat-Thin}
these two different expansions are called the fat expansion and
the thin expansion respectively and studied using
the results from \cite{Zhou-1D}.
In particular,
we apply the idea of renormalization of coupling constants
combined with Virasoro constraints
to study the structures of thin and fat free energy functions.

As noted in \cite{Zhou-Fat-Thin},
the distinction of fat and thin genus expansions leads to
the fat and thin versions of Virasoro constraints.
These are the point of departure for this paper.
Starting from \cite{Zhou-WK} we have seen that Virasoro constraints
in genus zero is equivalent to the emergence of spectral curve.
This idea was applied to the Hermitian matrix model with $N=1$
in \cite{Zhou-1D}.
So we have already known that the emergence of spectral curve
can occur also for finite $N$.
However,
the spectral curve discovered in \cite{Zhou-1D}  is not the spectral curve
that leads to the Wigner semicircle law
discovered in the limit of $N \to \infty$.
So for a long time our idea of using emergent geometry in the case of finite $N$
to study the emergent geometry for $N \to \infty$
has seemed to be only partially successful.
Even though we were able to develop emergent geometry for finite $N$
as in \cite{Zhou-1D},
the result was definitely different from the emergent geometry
for $N \to \infty$ in the literature.
So a puzzle arises as to unify these two different kinds of emergent geometry.
The solution to this problem naturally emerges
after we distinguish for matrix models
fat and thin correlators in \cite{Zhou-Mat-Mod},
and consequently distinguish fat and thin versions of free energy functions
and  their corresponding versions of Virasoro constraints
in \cite{Zhou-Fat-Thin}.
We will show that fat and thin free energy functions
in genus zero lead to different emergent geometries,
to be referred to as the fat and thin emergent geometries respectively.
Surprisingly,
they are even in some sense dual to each other.
But the point is that one of them (the fat one)
matches with the emergent geometry
in the $N\to \infty$ case in the literature,
the other (the thin one) generalizes the emergent geometry
of the case of $N=1$.
The thin spectral curve and the thin special deformation
turn out to be a deformation of the spectral curve and its
special deformation studied in \cite{Zhou-1D},
with $N$ as a parameter.
It is then natural that we can apply almost all results
in that paper to this situation.
For the fat spectral curve and the fat special deformation,
even though we work with finite $N$,
all the major techniques developed in the $N \to \infty$
setting turn out work as well.
So we have a perfect scenario here:
the thin emergent geometry is compatible with the $N=1$ case
and the fat emergent geometry is compatible with the $N \to \infty$ case.

Combinatorist  may have some interests in this work.
The thin special deformation is given by a formal power series
with {\em integral} coefficients,
and it is defined using enumerations of (thin) graphs.
The fat special deformation is also given by a formal power series
with integral coefficients,
but it is defined using enumerations of fat graphs.
They both lead to sophisticated generalizations of Catalan numbers.
It is a very remarkable fact discovered in this paper that the enumeration problems
of two different types of graphs can be unified
through the theory of matrix models.
In the concrete examples
to be presented below,
numerous integer sequences
on Sloane's Encyclopedia of Integer Sequences \cite{Sloane}
appear.
They include: A001764, A002293, A002294, A002295, A002296, A007556, A062994, A062744, A230388,
A104978, A002005, A001791, A002006-A002010, A000168, A085614, A000309.
Some of them even appear in more than one place.
They have various enumerative meanings.
Our work is in the spirit of \cite{BIZ}
in the sense that we want to apply string theoretical
idea to the combinatorial problem of enumerations of graphs,
with the difference being that we emphasize the possibility to
treat both fat and thin graphs simultaneously with the same techniques.

We have divided the rest of the paper into four Sections.
In Section \ref{sec:Thin} and Section \ref{sec:Fat}
we discuss the thin and fat versions of emergent geometries
of Hermitian one-matrix models respectively.
And in Section \ref{sec:Special}
we prove a very special property of fat special deformation as plane curves.
In these three Sections we have included numerous concrete
examples, they are not only used to illustrate the ideas,
but also to present some results that have their own interests.
For example,
we show that the emergence of spectral curves actually provide
a highly nontrivial trick to calculate some special correlators
indirectly using the Virasoro constraints.
The examples correspond to various combinatorial problems
considered in the literature separately.
It is interesting to see how they form special cases
of a unified picture.
We make some concluding remarks in Section \ref{sec:Conclusion}.

\section{Thin Emergent Geometry of Hermitian One-Matrix Models}
\label{sec:Thin}

In this Section we introduce the thin spectral curve and the thin special
deformation.
We also discuss various concrete examples.

\subsection{Virasoro constraints for thin genus expansion}
The thin Virasoro constraints are given by the following operators 
(cf. \cite[\S 3.4]{Zhou-Fat-Thin}):
\ben
&& L_{-1,N} = - \frac{\pd}{\pd g_1}
+ \sum_{n \geq 1} ng_{n+1} \frac{\pd}{\pd g_n} + Ng_1g_s^{-1}, \\
&& L_{0,N} = - 2\frac{\pd}{\pd g_2} + \sum_{n \geq 1} ng_n \frac{\pd}{\pd g_n}
+ N^2, \\
&& L_{1,N} = - 3\frac{\pd}{\pd g_3} + \sum_{n \geq 1} (n+1)g_n \frac{\pd}{\pd g_{n+1}}
+ 2Ng_s\frac{\pd}{\pd g_1}, \\
&& L_{m,N} = \sum_{k \geq 1} (k+m) (g_k - \delta_{k,2} ) \frac{\pd}{\pd g_{k+m}} +
g_s^2 \sum_{k=1}^{m-1} k(m-k)\frac{\pd}{\pd g_k} \frac{\pd}{\pd g_{m-k}} \\
&& \qquad\qquad + 2 Nm g_s \frac{\pd}{\pd g_m}, \qquad m \geq 2.
\een
In genus zero£¬
one has following equations for $F_{0,N}$:
\ben
&& \frac{\pd F_{0,N}}{\pd g_1}
= \sum_{n \geq 1} g_{n+1} \cdot n\frac{\pd F_{0,N}}{\pd g_n} + Ng_1, \\
&& (m+2)\frac{\pd F_{0,N}}{\pd g_{m+2}} = \sum_{k \geq 1}
g_k \cdot (k+m) \frac{\pd F_{0,N}}{\pd g_{k+m}}, \qquad m \geq 0.
\een
Therefore, if we set
\be
f_n =n\frac{\pd F_{0,N}}{\pd g_n}\biggl|_{g_j =0, j \geq 2},
\ee
then we have the following recursion relations and initial value£º
\begin{align}
f_1 & = N g_1, &
f_{n+1} & =  g_1f_{n}, \quad n \geq 1.
\end{align}
So we have
\be
n\frac{\pd F_{0,N}}{\pd g_n}\biggl|_{g_j =0, j \geq 2} = f_n = N g_1^n.
\ee

\subsection{Emergence of the signed Catalan curve}

If one introduce the following field  $y$ as a generating series of
the $f_n$'s:
\be \label{eqn:Thin-g1}
y = -\frac{z}{\sqrt{2}} + \frac{g_1}{\sqrt{2}} + \frac{\sqrt{2}N}{z} +\sqrt{2} \sum_{n \geq 1} \frac{f_n}{z^{n+1}},
\ee
then from the above formula for $f_n$ we have
\be \label{eqn:Thin-def-g1}
y = -\frac{z-g_1}{\sqrt{2}} + \sqrt{2}N \sum_{n\geq 0} \frac{g_1^n}{z^{n+1}}
= -\frac{z-g_1}{\sqrt{2}} + \frac{\sqrt{2}N}{z-g_1}.
\ee
This is a deformation of the following curve:
\be \label{eqn:Spec-Thin}
y = - \frac{1}{\sqrt{2}} z + \frac{\sqrt{2}N}{z},
\ee
which we call the {\em signed Catalan curve} in \cite{Zhou-1D}
in the case of $N=1$.
Note
\be \label{eqn:Signed-Catalan}
\frac{z}{\sqrt{2}} = \frac{-y+\sqrt{y^2+4N}}{2}
= \sum_{n=0}^\infty (-1)^nN^{n+1} \cdot \frac{1}{n+1}\binom{2n}{n} y^{-2n-1}.
\ee
The coefficients of the series on the right-hand side are Catalan numbers $\frac{1}{n+1}\binom{2n}{n}$
with signs $(-1)^n$,
and an extra factor $N^{n+1}$.

One can rewrite \eqref{eqn:Spec-Thin} and \eqref{eqn:Thin-def-g1}
in the form of algebraic curves as follows:
\bea
&& z^2 +\sqrt{2}y \cdot z + 2N = 0, \label{eqn:Spec-Thin2}\\
&& (z-g_1)^2 +\sqrt{2}y \cdot (z-g_1) + 2N = 0. \label{eqn:Thin-def-g12}
\eea

\subsection{Special deformations of the thin spectral curves}

We will refer to the curve in \eqref{eqn:Spec-Thin}
the {\em thin spectral curve} of the Hermitian matrix models.

The deformation of \eqref{eqn:Spec-Thin} in \eqref{eqn:Thin-def-g1}
defined by \eqref{eqn:Thin-g1} is just an example of special deformation deformation
defined by the field that appears in the derivation of Virasoro constraints.
See e.g. \cite[\S 3.2]{Zhou-Fat-Thin}.
Consider the field
\bea
\Phi_N(z) & = &  \frac{1}{\sqrt{2}} \sum_{n\geq 1} \tilde{T}_nz^n
- \sqrt{2}\tr \log \biggl( \frac{1}{z - M} \biggr) \\
& = & \frac{1}{\sqrt{2}} \sum_{n\geq 1} \tilde{T}_nz^n
+ \sqrt{2} N \log z
- \sqrt{2} \sum_{n\geq 1} \frac{z^{-n}}{n} \frac{\pd}{\pd T_n},
\eea
where $T_n = \frac{g_n}{n}= \frac{t_{n-1}}{n!}$, $\tilde{T}_n = \frac{g_n-\delta_{n,2}}{n}$,
and its derivative:
\be \label{def:y}
y :=  \frac{1}{\sqrt{2}} \sum_{n=1}^\infty n T_n z^{n-1}
+ \frac{\sqrt{2}N}{z} + \sqrt{2}
\sum_{n = 1}^\infty \frac{1}{z^{n+1}} \frac{\pd F_{0,N}}{\pd T_n}.
\ee
In \cite[\S ]{Zhou-Fat-Thin}
we have shown that
\be
F_{0,N} = N\cdot F_0^{1D},
\ee
where $F_0^{1D}$ is the genus zero free energy function of the topological 1D gravity
studied in \cite{Zhou-1D}.
The free energy $F^{1D}$ has extremely nice properties:
It satisfies the {\em flow equation} \cite[(262)]{Zhou-1D}
which in genus zero gives:
\be
\frac{\pd F_0^{1D}}{\pd t_m} = \frac{1}{(m+1)!} \biggl(\frac{\pd F_0^{1D}}{\pd t_0}\biggr)^{m+1}.
\ee
It also satisfies the {\em polymer equation} \cite[(153)]{Zhou-1D}
which in genus zero gives \cite[(160)]{Zhou-1D}:
\be \label{eqn:Polymer-Genus-0}
\frac{\pd F_0^{1D}}{\pd t_0}
= \sum_{n \geq 0} \frac{t_n}{n!} \biggl(\frac{\pd F_0^{1D}}{\pd t_0}\biggr)^n.
\ee
The solution of this equation is given in \cite[Theorem 5.3]{Zhou-1D} by:
\be
\frac{\pd F_0^{1D}}{\pd t_0} = I_0£¬
\ee
where $I_0$ is defined by:
\be \label{eqn:Xinfinity}
I_0 = \sum_{k=1}^\infty \frac{1}{k}
\sum_{p_1 + \cdots + p_k = k-1} \frac{t_{p_1}}{p_1!} \cdots
\frac{t_{p_k}}{p_k!}.
\ee

By combining the above results we then get the following

\begin{thm}
If the field $y$ is defined by \eqref{def:y},
then it is given by:
\be
y = \frac{1}{\sqrt{2}} \sum_{n=0}^\infty \frac{t_n -\delta_{n,1}}{n!} z^n
+ \frac{\sqrt{2}N}{z - I_0},
\ee
or equivalently,
it satisfies the following equation:
\be \label{eqn:Thin-Special-Def}
\sqrt{2} y(z - I_0) = (z - I_0) \sum_{n=0}^\infty \frac{t_n -\delta_{n,1}}{n!} z^n
+ 2N.
\ee
\end{thm}

We will refer to either of the two equations above as the {\em thin
special deformation} of \eqref{eqn:Spec-Thin} or \eqref{eqn:Spec-Thin2}.

\subsection{Characterization of the thin special deformation}

The following results are the generalizations of Theorem 10.1,
Theorem 10.2 and Theorem 10.3 in \cite{Zhou-1D} respectively,
which are the corresponding $N=1$ cases.
The proofs are exactly the same and so omitted.

\begin{thm} \label{thm:Thin-Existence}
Consider the following series:
\be
y = \frac{1}{\sqrt{2}} \sum_{n=0}^\infty \frac{t_n -\delta_{n,1}}{n!} z^n
+ \frac{\sqrt{2}N}{z} + \sqrt{2}
\sum_{n = 1}^\infty \frac{n!}{z^{n+1}} \frac{\pd F_{0,N}}{\pd t_{n-1}}.
\ee
Then one has:
\be
\half (y^2)_- =  \biggl(\sum_{n = 1}^\infty \frac{n!}{z^{n+1}}
\frac{\pd F_0}{\pd t_{n-1}}\biggr)^2.
\ee
Here for a formal series $\sum_{n \in \bZ} a_n f^n$,
\be
(\sum_{n \in \bZ} a_n f^n)_+ = \sum_{n \geq 0} a_n f^n, \;\;\;\;
(\sum_{n \in \bZ} a_n f^n)_ - = \sum_{n < 0} a_n f^n.
\ee
\end{thm}

\begin{thm} \label{thm:Uniqueness}
There exists a unique series
\be
y = \frac{1}{\sqrt{2}} \sum_{n \geq 0} (v_n-\delta_{n,1}) z^{n} + \frac{\sqrt{2}N}{z}
+ \sqrt{2} \sum_{n \geq 0} w_n z^{-n-2}
\ee
such that
each $w_n \in \bC[[v_0, v_1, \dots]]$ and
\be \label{eqn:y2-}
\half(y^2)_- =  \biggl(\frac{N}{z} +  \sum_{n \geq 0} w_n z^{-n-2}\biggr)^2.
\ee
\end{thm}

\begin{thm}
For a series of the form
\be
y =  \frac{1}{\sqrt{2}} \frac{\pd S(z, \bt)}{\pd z}
+ \frac{\sqrt{2}N}{z}+ \sqrt{2} \sum_{n \geq 0} w_n z^{-n-2},
\ee
where $S$ is the universal action defined by:
\be \label{eqn:Action}
S(x) = - \frac{1}{2}x^2 + \sum_{n \geq 1} t_{n-1} \frac{x^n}{n!},
\ee
and each $w_n \in \bC[[t_0, t_1, \dots]]$,
the equation
\be
\half (y^2)_- = \biggl(\frac{N}{z} +  \sum_{n \geq 0} w_n z^{-n-2}\biggr)^2
\ee
 has a unique solution given by:
\ben
&& y = \frac{1}{\sqrt{2}} \sum_{n=0}^\infty \frac{t_n -\delta_{n,1}}{n!} z^n
+ \frac{\sqrt{2}}{z} + \sqrt{2}
\sum_{n = 1}^\infty \frac{n!}{z^{n+1}} \frac{\pd F_{0,N}}{\pd t_{n-1}},
\een
where $F_{0,N}$ is the free energy of the Hermitian $N\times N$-matrix model.
\end{thm}

To gain some better understanding of the thin special deformations,
we will first examine some concrete examples in the next few Subsections.

\subsection{Thin special deformation on the $(g_1, g_2)$-plane}
\label{sec:Thin-g1g2}

By \eqref{eqn:Xinfinity},
\ben
I_0|_{t_k=0, k\geq 2} = \sum_{k \geq 1} t_0t_1^{k-1}= \frac{t_0}{1-t_1}
= \frac{g_1}{1-g_2},
\een
and so the special deformation of thin spectral curve on the $(g_1, g_2)$-space
is given by the following plane algebraic curve:
\be
\sqrt{2} y(z - \frac{t_0}{1-t_1}) = (z - \frac{t_0}{1-t_1})
(t_0 +(t_1-1)z)
+ 2N.
\ee
It can be rewritten in the following form:
\be \label{eqn:Thin-Def-g1g2}
\biggl(z - \frac{g_1}{1-g_2}\biggr)^2
+ \sqrt{2} \frac{y}{1-g_2}\biggl(z - \frac{g_1}{1-g_2}\biggr)
- \frac{2N}{1-g_2} = 0.
\ee
This is just  \eqref{eqn:Thin-def-g12} with the following changes of variables:
\begin{align*}
g_1 & \mapsto \frac{g_1}{1-g_2}, & y & \mapsto \frac{y}{1-g_2}, &
N & \mapsto \frac{2N}{1-g_2}.
\end{align*}

One can rewrite \eqref{eqn:Thin-Def-g1g2} as follows:
\be
\biggl((1-g_2)z - g_1\biggr)^2
+ \sqrt{2} y\biggl((1-g_2)z - g_1\biggr)
- 2N(1-g_2) = 0.
\ee
When $g_2=1$,
it becomes:
\be
y = \frac{g_1}{\sqrt{2}}.
\ee

\subsection{Thin special deformation along the $g_3$-line}
\label{sec:Thin-g3}
It is easy to see that
\ben
I_0|_{t_k=0, k \neq 2} = 0,
\een
and so the thin special deformation along the $g_3$-line is given by:
\ben
\sqrt{2} y \cdot z  = z \cdot (g_3z^2 -z)
+ 2N.
\een
It can be rewritten as:
\be
v = \frac{z}{2N - z^2(1 -g_3z)},
\ee
where $v= 1/(\sqrt{2}y)$.
From this one can use Lagrange inversion to express $z$ as a Taylor series in $v$
\ben
z= \sum_{n\geq 1} a_n v^n,
\een
as follows:
\ben
a_n & = & \frac{1}{2\pi i} \oint \frac{z}{v^{n+1}} dv \\
& = & \frac{1}{2\pi i} \oint \frac{z}{\biggl(\frac{z}{2N - z^2(1 -g_3z)}\biggr)^{n+1}}
d\frac{z}{2N - z^2(1 -g_3z)} \\
& = & \frac{1}{2\pi i} \oint \frac{(2N - z^2(1 -g_3z))^{n-1}
(2N+z^2-2g_3z^4)}{z^n}dz \\
& = & [2N - z^2(1 -g_3z))^{n-1}(2N+z^2-2g_3z^3)]_{z^{n-1}},
\een
where $[\cdot]_{z^{n-1}}$ means the coefficient of $z^{n-1}$.
Using the binomial theorem, it is possible to derive a closed formula for $a_{n}$
as follows.
For reason which will become clear later,
we will introduce an extra variable $h$, and proceed as follows:
\ben
&& (2N - w^2(h  -g_3w))^{n-1}  \\
& = & \sum_{a=0}^{n-1} (-1)^a \binom{n-1}{a}(2N)^{n-1-a} w^{2a} (h-g_4w)^a \\
& = & \sum_{a=0}^{n-1} (-1)^a \binom{n-1}{a}(2N)^{n-1-a} w^{2a}
\sum_{b=0}^a (-1)^b\binom{a}{b}h^{a-b}g_4^bw^{b} \\
& = & \sum_{0 \leq b \leq a \leq n-1} (-1)^{a+b} \frac{(n-1)!}{(n-1-a)!(a-b)!b!}
(2N)^{n-1-a}h^{a-b}g_4^bw^{2a+b},
\een
and then it follows that:
\ben
&& [(2N - w^2(h  -g_3w))^{n-1}(2N+hw^2-2g_3w^3)]_{w^{n-1}} \\
& = & 2N[(2N - w^2(h  -g_3w))^{2m}]_{w^{n-1}}
+[(2N - w^2(h  -g_3w))^{2m}]_{w^{n-3}} \\
& - & 2g_3[(2N - w^2(h  -g_3w))^{2m}]_{w^{n-4}} \\
& = & 2N\sum_{\substack{0 \leq b \leq a \leq n-1\\2a+b=n-1 }}
(-1)^{a+b} \frac{(n-1)!}{(n-1-a)!(a-b)!b!}
(2N)^{n-1-a}h^{a-b}g_4^b \\
& + &h \sum_{\substack{0 \leq b \leq a \leq n-1\\2a+b=n-3 }}
(-1)^{a+b} \frac{(n-1)!}{(n-1-a)!(a-b)!b!}
(2N)^{n-1-a}h^{a-b}g_4^b \\
& - & 2g_4\sum_{\substack{0 \leq b \leq a \leq n-1\\2a+b=n-4}}
(-1)^{a+b} \frac{(n-1)!}{(n-1-a)!(a-b)!b!}
(2N)^{n-1-a}h^{a-b}g_4^b  \\
& = & 2N \sum_{(n-1)/3 \leq a \leq (n-1)/2} (-1)^{n-1-a}
\frac{(n-1)!(2N)^{n-1-a}h^{3a-n+1}g_4^{n-1-2a} }{(n-1-a)!(n-1-2a)!(3a-n+1)!} \\
& + & h\sum_{(n-3)/3 \leq a \leq (n-3)/2} (-1)^{n-3-a}
\frac{(n-1)!(2N)^{n-1-a}h^{3a-n+3}g_4^{n-3-2a}}{(n-1-a)!(n-3-2a)!(3a-n+3)!} \\
& - &2 g_3 \sum_{(n-4)/3 \leq a \leq (n-4)/2} (-1)^{n-4-a}
\frac{(n-1)!(2N)^{n-1-a}h^{3a-n+4}g_4^{n-4-2a}}{(n-1-a)!(n-4-2a)!(3a-n+4)!} \\
& = & \sum_{(n-1)/3 \leq a \leq (n-1)/2} (-1)^{n-1-a}
\frac{(n-1)!(2N)^{n-1-a}h^{3a-n+1}g_4^{n-1-2a} }{(n-a)!(n-1-2a)!(3a-n+1)!},
\een
and so after taking $h=1$ we have:
\ben
z = \sum_{m\geq 0} \sum_{m/3 \leq a \leq m/2} (-1)^{m-a}
\frac{m!(2N)^{m-a} g_4^{m-2a} }{(m+1-a)!(m-2a)!(3a-m)!} v^{m+1}.
\een
The first few terms of $z$ are:
\be \label{eqn:z-g3}
\begin{split}
z = & (2N) v-(2N)^2v^3+(2N)^3g_3v^4+2(2N)^3v^5-5(2N)^4g_3v^6 \\
+ & (3 g_3^2 (2N)^5-5 (2N)^4)v^7+21 g_3  (2N)^5v^8 \\
+ & (-28 g_3^2 (2N)^6+14 (2N)^5)v^9+(12 g_3^3 (2N)^7-84 g_3 (2N)^6)v^{10} \\
+ & (180 g_3^2 (2N)^7-42 (2N)^6)v^{11} \\
+ & (-165 g_3^3 (2N)^8+330 g_3 (2N)^7)v^{12} \\
+ & (55 g_3^4 (2N)^9-990 g_3^2 (2N)^8+132 (2N)^7)v^{13} \\
+ & (1430 g_3^3 (2N)^9-1287 g_3 (2N)^8)v^{14}+ \cdots.
\end{split}
\ee
The coefficients $1$, $0$, $-1$, $1$, $2$,$-5$, $(3, -5)$, $21$, $(-28, 14)$,
$(12, -84)$, $(180,42)$, $\cdots$
are up to some signs and orders the sequence A104978 on \cite{Sloane}
given by the following formula:
\be
T(n,k) = \frac{(2n+k)!}{k!(n-k)!(n+k+1)!}, \qquad 1\leq k \leq n.
\ee

\subsection{Thin special deformation on the $(g_1, g_3)$-plane}
\label{sec:Thin-g1g3}
Next we turn on besides the coupling constant $g_3$ also  $g_1$.
Note
\ben
I_0|_{g_k=0, k \neq 1,3} & = & \sum_{k=1}^\infty \frac{1}{k}
\sum_{\substack{p_1 + \cdots + p_k = k-1\\ p_1, \dots, p_k \in \{0,2\}}}
g_{p_1+1} \cdots g_{p_k+1} \\
& = & \sum_{k=1}^\infty \frac{1}{k} \sum_{\substack{m_0+m_2=k\\2m_2=k-1}}
\frac{k!}{m_0!m_2!} g_1^{m_0} g_3^{m_2} \\
& = & \sum_{m=0}^\infty \frac{1}{2m+1} \cdot \frac{(2m+1)!}{(m+1)!m!} g_1^{m+1}g_3^m \\
& = & g_1 \sum_{m=0}^\infty \frac{1}{m+1} \binom{2m}{m} g_1^{m}g_3^m,
\een
the coefficients are Catalan numbers,
and so
\be
I_0|_{g_k=0, k \neq 1,3} = g_1 \cdot \frac{1- \sqrt{1-4g_1g_3}}{2g_1g_3}.
\ee
Therefore,
the thin special deformation on $(g_1, g_3)$-plane is given by the equation:
\be \label{eqn:Thin-Def-g1g3}
\sqrt{2} y\biggl(z - \frac{1- \sqrt{1-4g_1g_3}}{2g_3} \biggr)
= \biggl(z - \frac{1- \sqrt{1-4g_1g_3}}{2g_3} \biggr)(g_1 -z+g_3z^2)
+ 2N,
\ee

If we make the following change of coordinates:
\begin{align*}
w & = z - \frac{1- \sqrt{1-4g_1g_3}}{2g_3}, &
v & = \frac{1}{\sqrt{2}y},
\end{align*}
then we get the following equation:
\ben
w = 2N v +v w \biggl[g_1 - \frac{1- \sqrt{1-4g_1g_3}}{2g_3}
- w
+ g_3 \biggl(w +\frac{1- \sqrt{1-4g_1g_3}}{2g_3}  \biggr)^2 \bigg].
\een
After simplification,
\be
w = 2N v -vw^2(\sqrt{1-4g_1g_3} -g_3w),
\ee
or equivalently,
\be
v = \frac{w}{2N - w^2(\sqrt{1-4g_1g_3} -g_3w)}
\ee
From this one can use Lagrange inversion to solve for $w$ as in last Subsection to get:
\ben
w= \sum_{m\geq 0} \sum_{m/3 \leq a \leq m/2} (-1)^{m-a}
\frac{m!(2N)^{m-a}h^{3a-m}g_4^{m-2a} }{(m+1-a)!(m-2a)!(3a-m)!} v^{m+1},
\een
where $h = \sqrt{1-4g_1g_3}$.
By writing down the first few terms of $w$ explicitly we get:
\ben
z & = & \frac{1-h}{2g_3} + (2N) v-(2N)^2hv^3+(2N)^3g_3v^4+2(2N)^3h^2v^5-5(2N)^4g_3hv^6 \\
& + & (3 g_3^2 (2N)^5-5 h^3 (2N)^4)v^7+21 g_3 h^2 (2N)^5v^8 \\
& + & (-28 g_3^2 h (2N)^6+14 h^4 (2N)^5)v^9+(12 g_3^3 (2N)^7-84 g_3 h^3 (2N)^6)v^{10} \\
& + & (180 g_3^2 h^2 (2N)^7-42 h^5 (2N)^6)v^{11} \\
& + & (-165 g_3^3 h (2N)^8+330 g_3 h^4 (2N)^7)v^{12} \\
& + & (55 g_3^4 (2N)^9-990 g_3^2 h^3 (2N)^8+132 h^6 (2N)^7)v^{13} \\
& + & (1430 g_3^3 h^2 (2N)^9-1287 g_3 h^5 (2N)^8)v^{14}+ \cdots,
\een
from this one sees explicitly how one can first deform along the $g_3$-line,
then deform along the $g_1$-direction.

\subsection{Thin special deformation on the $(g_1,g_2, g_3)$-space}
\label{sec:Thin-g1g2g3}
Note
\ben
I_0|_{t_k=0, k \geq 3} & = & \sum_{k=1}^\infty \frac{1}{k}
\sum_{\substack{p_1 + \cdots + p_k = k-1\\0 \leq p_1, \dots, p_k \leq 2}}
g_{p_1+1} \cdots g_{p_k+1} \\
& = & \sum_{k=1}^\infty \frac{1}{k} \sum_{\substack{m_0+m_1+m_2=k\\m_1+2m_2=k-1}}
\frac{k!}{m_0!m_1!m_2!} g_1^{m_0}g_2^{m_1} g_3^{m_2}.
\een
From the system
\ben
&& m_0+m_1+m_2=k, \\ && m_1+2m_2=k-1,
\een
we get $m_0 = m_2 +1$ and so $k = m_1+2m_2+1$, $m_1$ can be arbitrary,
it follows that
\ben
I_0|_{t_k=0, k \geq 3}
& = & \sum_{m_2=0}^\infty \sum_{m_1=0}^\infty \frac{1}{m_1+2m_2+1}
\frac{(m_1+2m_2+1)!}{(m_2+1)!m_1!m_2!} g_1^{m_2+1}g_2^{m_1} g_3^{m_2} \\
& = & \sum_{m_2\geq 0} \frac{(2m_2)!}{m_2!(m_2+1)!}
\frac{g_1^{m_2+1}  g_3^{m_2} }{(1-g_2)^{2m_2+1}} \\
& = & \frac{1- \sqrt{1-4\frac{g_1}{1-g_2}\frac{g_3}{1-g_2}}}{2\frac{g_3}{1-g_2}},
\een
and so the special deformation on the $(g_1,g_2, g_3)$-space is given by:
\be \label{eqn:Thin-Def-g1g2g3}
\begin{split}
& \sqrt{2} \frac{y}{1-g_2} \biggl(z - \frac{1-
\sqrt{1-4\frac{g_1}{1-g_2}\frac{g_3}{1-g_2}}}{2\frac{g_3}{1-g_2}}\biggr) \\
= &  \biggl(z - \frac{1- \sqrt{1-4\frac{g_1}{1-g_2}\frac{g_3}{1-g_2}}}{2\frac{g_3}{1-g_2}}
\biggr)\biggl(\frac{g_1}{1-g_2} -z + \frac{g_3}{1-g_2} z^2 \biggr) + \frac{2N}{1-g_2},
\end{split}
\ee
This is just \eqref{eqn:Thin-Def-g1g3} with the following changes of variables:
\begin{align*}
g_1 & \mapsto \frac{g_1}{1-g_2}, &
g_3 & \mapsto \frac{g_3}{1-g_2}, &
y & \mapsto \frac{y}{1-g_2}, &
N & \mapsto \frac{2N}{1-g_2}.
\end{align*}

One can rewrite \eqref{eqn:Thin-Def-g1g2g3} as follows:
\be
\begin{split}
& \sqrt{2}  y  \biggl(z - \frac{1-g_2-
\sqrt{(1-g_2)^2-4 g_1 g_3 }}{2 g_3 }\biggr) \\
= &  \biggl(z - \frac{1-g_2- \sqrt{(1-g_2)^2-4g_1 g_3 }}{2g_3 }
\biggr)\biggl(g_1- (1-g_2)z + g_3  z^2 \biggr) + 2N ,
\end{split}
\ee
and so when $g_2=1$,
\be
 \sqrt{2}  y  \biggl(z + \frac{
\sqrt{-g_1 g_3 }}{g_3 }\biggr)
= \biggl(z + \frac{\sqrt{-g_1 g_3 }}{g_3 }
\biggr)\biggl(g_1 + g_3  z^2 \biggr) + 2N.
\ee
The appearance of fractional powers of the coupling constants
indicates the occurrence of some phase transition.
If we furthermore take $g_1=0$,
then we get:
\be
\sqrt{2}y \cdot z = g_3 z^3 +2N.
\ee

\subsection{Thin special deformation on the $g_4$-line} \label{sec:Thin-g4}
In this case  the thin special deformation is given by
\be
\sqrt{2} y z= -z^2(1-g_4z^2)+2N.
\ee
To express $z$ as a Taylor series in $1/y$,
we need to solve:
\be
v = \frac{z}{2N - z^2(1  -g_4z^2)}.
\ee
We use Lagrange inversion to solve for $z$ as follows.
Write
\ben
z= \sum_{n\geq 1} a_n v^n,
\een
then we have
\ben
a_n & = & \frac{1}{2\pi i} \oint \frac{z}{v^{n+1}} dv \\
& = & \frac{1}{2\pi i} \oint \frac{z}{\biggl(\frac{z}{2N - z^2(1  -g_4z^2)}\biggr)^{n+1}}
d\frac{z}{2N - z^2(1  -g_4z^2)} \\
& = & \frac{1}{2\pi i} \oint \frac{(2N - z^2(1  -g_4z^2))^{n-1}(2N+z^2-3g_4z^4)}{z^n}dz \\
& = & [(2N - z^2(1  -g_4z^2))^{n-1}(2N+z^2-3g_4z^4)]_{z^{n-1}}.
\een
It is clear that $a_{2m}$ all vanish.
Using the binomial theorem, it is possible to derive a closed formula for $a_{2m+1}$
as follows.
We first get:
\ben
&& (2N - z^2(1  -g_4z^2))^{2m}  \\
& = & \sum_{a=0}^{2m} (-1)^a \binom{2m}{a}(2N)^{2m-a} z^{2a} (1-g_4z^2)^a \\
& = & \sum_{a=0}^{2m} (-1)^a \binom{2m}{a}(2N)^{2m-a} z^{2a}
\sum_{b=0}^a (-1)^b\binom{a}{b}g_4^bz^{2b} \\
& = & \sum_{0 \leq b \leq a \leq 2m} (-1)^{a+b} \frac{(2m)!}{(2m-a)!(a-b)!b!}
(2N)^{2m-a}g_4^bz^{2(a+b)},
\een
and then it follows that:
\ben
&& [(2N - z^2(1  -g_4z^2))^{2m}(2N+z^2-3g_4z^4)]_{z^{2m}} \\
& = & 2N[(2N - z^2(1  -g_4z^2))^{2m}]_{z^{2m}}
+[(2N - z^2(1  -g_4z^2))^{2m}]_{z^{2m-2}} \\
& - & 3g_4[(2N - z^2(1  -g_4z^2))^{2m}]_{z^{2m-4}} \\
& = & 2N\sum_{\substack{0 \leq b \leq a \leq 2m\\a+b=m}}
(-1)^{a+b} \frac{(2m)!}{(2m-a)!(a-b)!b!} (2N)^{2m-a}g_4^b \\
& + & \sum_{\substack{0 \leq b \leq a \leq 2m\\a+b=m-1}} (-1)^{a+b}
\frac{(2m)!}{(2m-a)!(a-b)!b!} (2N)^{2m-a}g_4^b \\
& - & 3g_4\sum_{\substack{0 \leq b \leq a \leq 2m\\a+b=m-2}} (-1)^{a+b}
\frac{(2m)!}{(2m-a)!(a-b)!b!} (2N)^{2m-a}g_4^b \\
& = & \sum_{b=0}^{[m/2]} (-1)^m \frac{(2m)!}{b!(m-2b)!(m+b)!}
(2N)^{m+b+1}g_4^b \\
& + & \sum_{b=0}^{[(m-1)/2]} (-1)^{m-1} \frac{(2m)!}{b!(m-1-2b)!(m+1+b)!}
(2N)^{m+b+1}g_4^b \\
& - &3 g_4 \sum_{b=0}^{[(m-2)/2]} (-1)^{m-2} \frac{(2m)!}{b!(m-2-2b)!(m+2+b)!}
(2N)^{m+b+2}g_4^b.
\een
After simplification we get:
\be
a_{2m+1} = (-1)^m \sum_{b=0}^{[m/2]}
\frac{(2m)!}{b!(m-2b)!(m+1+b)!} (2N)^{m+b+1}g_4^b,
\ee
and so we have:
\ben
z= \sum_{m\geq 0} (-1)^m \sum_{b=0}^{[m/2]}
\frac{(2m)!}{b!(m-2b)!(m+1+b)!} (2N)^{m+b+1}g_4^b v^{2m+1}.
\een
The following are the first few terms:
\ben
z & = & (2N) v-(2N)^2v^3+((2N)^4g_4+2(2N)^3)v^5  \\
& + & (6 g_4 (2N)^5+5(2N)^4)v^7+(4g_4^2(2N)^7+28g_4(2N)^6+14(2N)^5)v^9\\
& - & (45g_4^2(2N)^8+120g_4(2N)^7+42(2N)^6)v^{11} \\
& + & (22g_4^3(2N)^{10}+330g_4^2(2N)^9+495g_4(2N)^8+132(2N)^7)v^{13} \\
& - & (364g_4^3(2N)^{11}+2002g_4^2(2N)^{10}+2002g_4(2N)^9+429(2N)^8)v^{15}+ \cdots,
\een
the coefficients (up to signs) $1$, $1$, $(1, 2)$, $(6,5)$, $(4,28,14)$,
$(45,120,42)$, $(22$,$330$,$495$, $132)$, $\cdots$ are generalizations of
the Catalan numbers $1,1,2,5,14,42,132,\cdots$,
but they have not yet appeared on \cite{Sloane}.

\subsection{Thin special deformation on the $(g_1, g_4)$-plane}
\label{sec:Thin-g1g4}

In this case
\ben
I_0|_{t_k=0, k \neq 0,3}
& = & \sum_{k=1}^\infty \frac{1}{k} \sum_{\substack{m_0+m_3=k\\3m_3=k-1}}
\frac{k!}{m_0!m_3!} g_1^{m_0} g_4^{m_3} \\
& = & \sum_{m=0}^\infty \frac{1}{3m+1} \cdot \frac{(3m+1)!}{(2m+1)!m!} g_1^{2m+1}g_4^m \\
& = & g_1 \sum_{m=0}^\infty \frac{1}{2m+1} \binom{3m}{m} g_1^{2m}g_4^m,
\een
the coefficients are the sequence A001764 on \cite{Sloane},
they are the numbers of complete ternary trees with $n$ internal nodes.
Write
\be
A(x) = \sum_{m=0}^\infty \frac{1}{2m+1} \binom{3m}{m}x^m,
\ee
then one has \cite[A001764]{Sloane}:
\be \label{eqn:A}
A(x) = 1+xA(x)^3 = \frac{1}{1-xA(x)^2}.
\ee
The thin special deformation is given by the plane algebraic curve:
\be
\sqrt{2} y\biggl(z - g_1A(g_1^2g_4) \biggr)
= \biggl(z - g_1A(g_1^2g_4)\biggr)(g_1 -z+g_4z^3)
+ 2N,
\ee
Make the following change of coordinates:
\begin{align*}
w & = z - g_1A(g_1^2g_4), &
v &= \frac{1}{\sqrt{2}y}.
\end{align*}
Then we get the following equation:
\ben
w = 2N v +v w \biggl[g_1 - g_1A(g_1^2g_4)- w
+ g_4 \big(w +g_1A(g_1^2g_4)\big)^3 \bigg].
\een
After simplification using \eqref{eqn:A},
we get
\be
w = 2N v -vw^2(1-3g_1^2g_4A^2(g_1^2g_4)-3g_1g_4A(g_1^2g_4)w -g_4w^2),
\ee
or equivalently,
\be
v = \frac{w}{2N - w^2(1-3g_1^2g_4A^2(g_1^2g_4)-3g_1g_4A(g_1^2g_4)w -g_4w^2)}.
\ee
From this one can use Lagrange inversion to solve for $w$.

\subsection{Thin special deformation on the $g_{k+2}$-line ($k \geq 3$)}
It is straightforward to generalize the results in \S \ref{sec:Thin-g3} and 
\S \ref{sec:Thin-g4}.
The thin special deformation along the $g_{k+2}$-line is given by
\be
\sqrt{2} y z= -z^2+g_{k+2}z^{k+2}+2N.
\ee
To express $z$ as a Taylor series in $1/y$,
we need to solve:
\be
v = \frac{z}{2N - z^2+g_{k+2}z^{k+2}}.
\ee
We use Lagrange inversion to solve for $z$ as follows.
Write
\ben
z= \sum_{n\geq 1} a_n v^n,
\een
then we have
\ben
a_n & = & \frac{1}{2\pi i} \oint \frac{z}{v^{n+1}} dv \\
& = & \frac{1}{2\pi i} \oint \frac{z}{\biggl(\frac{z}{2N - z^2(1  -g_{k+2}z^k)}\biggr)^{n+1}}
d\frac{z}{2N - z^2(1  -g_{k+2}z^k)} \\
& = & \frac{1}{2\pi i} \oint \frac{(2N - z^2(1  -g_{k+2}z^k))^{n-1}(2N+z^2-(k+1)g_{k+2}z^{k+2})}{z^n}dz \\
& = & [(2N - z^2(1  -g_{k+2}z^k))^{n-1}(2N+z^2-(k+1)g_{k+2}z^{k+2})]_{z^{n-1}}.
\een
Using the binomial theorem, it is possible to derive a closed formula for $a_{2m+1}$
as follows.
We first get:
\ben
&& (2N - z^2(1  -g_{k+2}z^k))^{n-1}  \\
& = & \sum_{a=0}^{n-1} (-1)^a \binom{n-1}{a}(2N)^{n-1-a} z^{2a} (1-g_{k+2}z^k)^a \\
& = & \sum_{a=0}^{n-1} (-1)^a \binom{n-1}{a}(2N)^{n-1-a} z^{2a}
\sum_{b=0}^a (-1)^b\binom{a}{b}g_{k+2}^bz^{kb} \\
& = & \sum_{0 \leq b \leq a \leq n-1} (-1)^{a+b} \frac{(n-1)!}{(n-1-a)!(a-b)!b!}
(2N)^{n-1-a}g_{k+2}^bz^{2a+kb},
\een
and then it follows that:
\ben
&& [(2N - z^2(1  -g_{k+2}z^k))^{n-1}(2N+z^2-(k+1)g_{k+2}z^{k+2})]_{z^{n-1}} \\
& = & 2N[(2N - z^2(1  -g_{k+2}z^k))^{n-1}]_{z^{n-1}}
+[(2N - z^2(1  -g_{k+2}z^k))^{n-1}]_{z^{n-3}} \\
& - & (k+1)g_4[(2N - z^2(1  -g_{k+2}z^k))^{n-1}]_{z^{n-k-3}} \\
& = & 2N\sum_{\substack{0 \leq b \leq a \leq n-1\\2a+kb=n-1}}
(-1)^{a+b} \frac{(n-1)!}{(n-1-a)!(a-b)!b!} (2N)^{n-1-a}g_{k+2}^b \\
& + & \sum_{\substack{0 \leq b \leq a \leq n-1\\2a+kb=n-3}}
(-1)^{a+b} \frac{(n-1)!}{(n-1-a)!(a-b)!b!} (2N)^{n-1-a}g_{k+2}^b \\
& - & (k+1) g_{k+2}\sum_{\substack{0 \leq b \leq a \leq n-1\\2a+kb=n-k-3}}
(-1)^{a+b} \frac{(n-1)!}{(n-1-a)!(a-b)!b!} (2N)^{n-1-a}g_{k+2}^b \\
& = & \sum_{\substack{(n-1)/(k+2) \leq a \leq (n-1)/2\\
0 \leq b =(n-1-2a)/k} } (-1)^{a+b} \frac{(n-1)!}{(n-1-a)!(a-b)!b!}
(2N)^{n-a}g_{k+2}^b \\
& + & \sum_{\substack{(n-1)/(k+2) \leq a \leq (n-1)/2\\
0 \leq b =(n-3-2a)/k} } (-1)^{a+b} \frac{(n-1)!}{(n-1-a)!(a-b)!b!}
(2N)^{n-1-a}g_{k+2}^b \\
& - & (k+1) g_{k+2} \sum_{\substack{(n-1)/(k+2) \leq a \leq (n-1)/2\\
0 \leq b =(n-3-2a)/k-1} } (-1)^{a+b} \frac{(n-1)!}{(n-1-a)!(a-b)!b!}
(2N)^{n-1-a}g_{k+2}^b \\
& = & \sum_{\substack{(n-1)/(k+2) \leq a \leq (n-1)/2\\
0 \leq b =(n-1-2a)/k} } (-1)^{a+b} \frac{(n-1)!}{(n-a)!(a-b)!b!}
(2N)^{n-a}g_{k+2}^b,
\een
and so we have:
\ben
z= \sum_{m\geq 0} v^{m+1} \sum_{\substack{m/(k+2) \leq a \leq m/2\\
0 \leq b =(m-2a)/k} } (-1)^{a+b} \frac{m!}{(m+1-a)!(a-b)!b!}
(2N)^{m+1-a}g_{k+2}^b.
\een
The   coefficients (up to signs)   are all generalizations of
the Catalan numbers.

\subsection{Thin special deformation on the $(g_1, g_{k+2})$-plane}

In this case
\ben
I_0|_{g_j=0, j \neq 1,k+2}
& = & \sum_{l=1}^\infty \frac{1}{l} \sum_{\substack{m_1+m_{k+2}=l\\(k+1) m_{k+2} =l-1}}
\frac{l!}{m_0!m_3!} g_1^{m_0} g_4^{m_3} \\
& = & \sum_{m=0}^\infty \frac{1}{(k+1)m+1} \cdot \frac{((k+1)m+1)!}{(km+1)!m!} g_1^{km+1}g_{k+2}^m \\
& = & g_1 \sum_{m=0}^\infty \frac{1}{km+1} \binom{(k+1)m}{m} g_1^{km}g_{k+2}^m,
\een
the coefficients are various sequences  on \cite{Sloane},
they enumerate $(k+1)$-ary rooted trees  with $m$ internal nodes.
Write
\be
A_k(x) = \sum_{m=0}^\infty \frac{1}{km+1} \binom{(k+1)m}{m}x^m,
\ee
then one has:
\be \label{eqn:A-k}
A_k(x) = 1+xA(x)^{k+1}.
\ee
The thin special deformation is given by the plane algebraic curve:
\be
\sqrt{2} y\biggl(z - g_1A_k(g_1^kg_{k+2}) \biggr)
= \biggl(z - g_1A_k(g_1^kg_{k+2})\biggr)(g_1 -z+g_{k+2}z^{k+1})
+ 2N,
\ee
As before,
after we make the following change of coordinates:
\begin{align*}
w & = z - g_1A_k(g_1^kg_4), &
v &= \frac{1}{\sqrt{2}y},
\end{align*}
we get the following equation:
\ben
w = 2N v +v w \biggl[g_1 - g_1A_k(g_1^kg_{k+2})- w
+ g_{k+2} \cdot \big(w +g_1A_k(g_1^kg_{k+2})\big)^{k+1} \bigg].
\een
After simplification using \eqref{eqn:A-k},
we get
\be
w = 2N v +vw^2 \biggl(g_{k+2}\sum_{j}\binom{k+1}{j}
 (g_1A_k(g_a^kg_{k+2}))^j w^{k-j} -1\biggr).
\ee
After rewriting it in the following form:
\be
v = \frac{w}{2N + w^2 \biggl(g_{k+2}\sum_{j}\binom{k+1}{j}
 (g_1A_k(g_a^kg_{k+2}))^j w^{k-j} -1\biggr)}.
\ee
one can use Lagrange inversion to solve for $w$
as a poser series in $v$.

\subsection{Thin special deformation in renormalized coupling constants}

To understand the examples discussed in the above several Subsections,
it is important to turn on all the coupling constants
and to consider the problem of expressing $z$ as
a Taylor series in $v= \frac{1}{\sqrt{2}y}$.
It turns out to be convenient to
change to the renormalized coupling constants $I_k$:
\bea
&& I_0 = \sum_{k=1}^\infty \frac{1}{k}
\sum_{p_1 + \cdots + p_k = k-1} \frac{t_{p_1}}{p_1!} \cdots
\frac{t_{p_k}}{p_k!}, \label{eqn:I0} \\
&& I_k= \sum_{n \geq 0} t_{n+k} \frac{I_0^n}{n!}, \;\;\;\; k \geq 1. \label{eqn:Ik}
\eea
These were used in \cite{Zhou-1D} to rewrite the action function:
\be
\begin{split}
S(z) & = \sum_{n \geq 1} (t_{n-1} -\delta_{n,2}) \frac{x^n}{n!} \\
& = \sum_{k=0}^\infty  \frac{(-1)^k}{(k+1)!} (I_k+\delta_{k,1}) I_0^{k+1}
+ \sum_{n \geq 2} (I_{n-1} - \delta_{n,2}) \frac{(z-I_0)^n}{n!}.
\end{split}
\ee
Because the thin special deformation is given by:
\be
y = \frac{1}{\sqrt{2}} S'(z) +\frac{\sqrt{2}N}{z-I_0},
\ee
so in the renormalized coupling constant
it can be written as
\be
y = \frac{1}{\sqrt{2}}\sum_{n \geq 1} (I_n-\delta_{n,1}) \frac{(z-I_0)^n}{n!}
+ \frac{\sqrt{2}N}{z-I_0}.
\ee
Make the following changes of variables:
\begin{align}
w & = z- I_0, & v & = \frac{1}{\sqrt{2}y},
\end{align}
then we get:
\be \label{eqn:Thin-Lag}
v = \frac{w}{2N + \sum_{n \geq 1} (I_n-\delta_{n,1}) \frac{w^{n+1}}{n!}}.
\ee
From this one can use Lagrange inversion to express $w$ as a
Taylor series of $v$, with coefficients polynomials in $2N$ and $I_n - \delta_{n,1}$.

Let us now explain some of the combinatorics involved here.
They are all based on an earlier work \cite{Zhou-1D} to which we refer for more details and
proofs.
First of all,
the renormalized coupling constant $I_0$, defined by the explicit formula \eqref{eqn:I0} above,
with the first few terms explicitly given by
\ben
I_0 & = &  t_0 + t_1t_0 +
\biggl( t_2\frac{t^2_0}{2!} + 2\frac{t^2_1}{2!} t_0 \biggr)
+ \biggl(t_3\frac{t^3_0}{3!} + 3t_1t_2\frac{t^2_0}{2!}
+ 6 \frac{t^3_1}{3!}  t_0 \biggr) \\
& + & \biggl[t_4\frac{t^4_0}{4!}
+ \biggl( 6\frac{t_2^2}{2!} + 4t_1t_3\biggr) \frac{t^3_0}{3!}
+ 12t_2 \frac{t^2_1}{2!} \frac{t^2_0}{2!} + 24 \frac{t^4_1}{4!}  t_0 \biggl] \\
& + &  \biggl[ t_5\frac{t^5_0}{5!} + (5t_1t_4 + 10t_2t_3) \frac{t^4_0}{4!}
+ \biggl( 30t_1 \frac{t_2^2}{2!} + 20t_3\frac{t^2_1}{2!} \biggl) \frac{t^3_0}{3!} \\
&& + 60t_2\frac{t^3_1}{3!} \frac{t^2_0}{2!} + 120 \frac{t^5_1}{5!} t_0 \biggr] \\
& + & \biggl[ t_6\frac{t^6_0}{6!}  +
\biggl( 20 \frac{t^2_3}{2!} + 6t_1t_5 + 15t_2t_4 \biggr) \frac{t^5_0}{5!} \\
&& + \biggl( 90 \frac{t^3_2}{3!} + 30t_4 \frac{t^2_1}{2!} + 60t_1t_2t_3
\biggr)\frac{t^4_0}{4!}  \\
&& +  \biggl( 120t_3 \frac{t^3_1}{3!} + 180 \frac{t^2_1}{2!} \frac{t^2_2}{2!}  \biggr)
\frac{t^3_0}{3!}
+ 360t_2 \frac{t^4_1}{4!} \frac{t^2_0}{2!} + 720  \frac{t^6_1}{6!} t_0
\biggr) + \cdots,
\een
has a combinatorial interpretation of an enumeration of
rooted trees given by Feynman rules to be specified below.
By a rooted tree we mean a tree whose
vertices are all marked by $\bullet$,
except for a valence-one vertex marked by $\circ$.
This exceptional vertex was referred to as the {\em root vertex $\circ$} in
\cite{Zhou-1D}.
This is slightly different from the standard notion of a rooted tree in the literature,
which means just a tree with a specified vertex, called the root of the tree.
Our version of the rooted tree is obtained from the standard version
by attaching an edge to the root, with the other vertex marked by $\circ$.
With this understood, we can recall the Feynman rules for $I_0$ \cite[Theorem 3.1]{Zhou-1D}:
\be \label{eqn:I0-Feynman-1}
I_0 = \sum_{\text{$\Gamma$ is a rooted tree}} \frac{1}{|\Aut \Gamma|} w_\Gamma,
\ee
where the weight of $\Gamma$ is given by
\be \label{eqn:I0-Feynman-2}
w_\Gamma = \prod_{v\in V(\Gamma)} w_v \cdot \prod_{e\in E(\Gamma)} w_e,
\ee
with $w_e$ and $w_v$ given by the following Feynman rule:
\be \label{eqn:I0-Feynman-3}
\begin{split}
& w(e) = 1, \\
& w(v) = \begin{cases}
t_{\val(v)-1}, & \text{if $v$ is not the root vertex $\circ$}, \\
1, & \text{if $v$ is the root vertex $\circ$}.
\end{cases}
\end{split}
\ee
For example,
$$
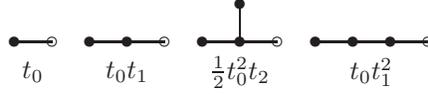
\xy
(0,0); (5,0), **@{-}; (0, 0)*+{\bullet}; (5,0)*+{\circ}; (2.5,-4)*+{t_0};
(10,0); (20,0), **@{-};  (10, 0)*+{\bullet}; (15, 0)*+{\bullet}; (20,0)*+{\circ}; (15,-4)*+{t_0t_1};
(25,0); (35,0), **@{-};  (30,0); (30,5), **@{-};
(25, 0)*+{\bullet}; (30, 0)*+{\bullet}; (30, 5)*+{\bullet}; (35,0)*+{\circ}; (30,-4)*+{\half t_0^2t_2};
(40,0); (55,0), **@{-};  (40, 0)*+{\bullet}; (45, 0)*+{\bullet}; (50, 0)*+{\bullet}; (55,0)*+{\circ};
(47.5,-4)*+{t_0t_1^2};
\endxy
$$
give the first few terms of $I_0$.
Secondly, all the other renormalized coupling constants
$I_k$ have similar combinatorial interpretations.
We need to extend our version of the rooted trees to a notion of {\em rooted tree of type $k$}.
This means a tree that is obtained from a standard rooted tree by
attaching $k+1$ edges to the root vertex,
and mark each of the extra $k+1$ vertices by $\circ$.
The renormalized coupling constant $I_k$ is given by a sum over rooted trees of type $k$:
\be
\frac{1}{k!}I_k = \sum_{\text{$\Gamma$ is a rooted tree of type $k$}} \frac{1}{|\Aut \Gamma|} w_\Gamma,
\ee
where the weight of $\Gamma$ is given by
\be
w_\Gamma = \prod_{v\in V(\Gamma)} w_v \cdot \prod_{e\in E(\Gamma)} w_e,
\ee
with $w_e$ and $w_v$ given by the following Feynman rule:
\bea
&& w(e) = 1, \\
&& w(v) = \begin{cases}
t_{\val(v)-1}, & \text{if $v$ is not a vertex marked by $\circ$}, \\
1, & \text{if $v$ is a vertex marked by $\circ$}  .
\end{cases}
\eea
For example,
$$
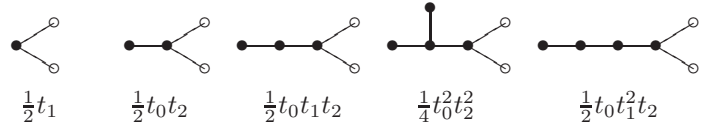
\xy
(-10,3); (-15,0), **@{-}; (-10,-3), **@{-}; (-10,3)*+{\circ}; (-15,0)*+{\bullet}; (-10,-3)*+{\circ};
(-12,-8)*+{\half t_1};
(0,0); (5,0), **@{-}; (10,3);  (5,0), **@{-};(10,-3), **@{-};
(0, 0)*+{\bullet}; (5,0)*+{\bullet}; (10,3)*+{\circ}; (10,-3)*+{\circ};
(4,-8)*+{\half t_0t_2};
(15,0); (25,0), **@{-}; (30,3); (25,0), **@{-}; (30,-3), **@{-};
(15, 0)*+{\bullet}; (20, 0)*+{\bullet}; (25,0)*+{\bullet}; (30,3)*+{\circ}; (30,-3)*+{\circ};
(23,-8)*+{\half t_0t_1t_2};
(35,0); (45,0), **@{-};  (40,0); (40,5), **@{-}; (50,3); (45,0), **@{-}; (50,-3), **@{-};
(35, 0)*+{\bullet}; (40, 0)*+{\bullet}; (40, 5)*+{\bullet};
(45,0)*+{\bullet}; (50,3)*+{\circ}; (50,-3)*+{\circ};
(42,-8)*+{\frac{1}{4} t_0^2t_2^2};
(55,0); (70,0), **@{-}; (75,3); (70,0), **@{-}; (75,-3), **@{-};
(55, 0)*+{\bullet}; (60, 0)*+{\bullet}; (65, 0)*+{\bullet};
(70,0)*+{\bullet}; (75,3)*+{\circ}; (75,-3)*+{\circ};
(65,-8)*+{\half t_0 t_1^2 t_2};
\endxy
$$
give the first few terms of $\frac{1}{2!}I_1$.

Thirdly,
taking the Lagrange inversion also has both a closed formula and
a combinatorial interpretation exactly as in the case of $I_0$,
in fact by the results in \cite{Zhou-1D},
the Lagrange inversion of
\be
z = \frac{w}{ J_0 + \sum_{n \geq 1} J_n \frac{w^n}{n!}}
\ee
is solved by the following explicit formula:
\be
w = \sum_{k=1}^\infty \frac{z^k}{k}
\sum_{p_1 + \cdots + p_k = k-1} \frac{J_{p_1}}{p_1!} \cdots
\frac{J_{p_k}}{p_k!}.
\ee
This is proved in \cite[Proposition 2.2]{Zhou-1D}
with the only change $t_n \to J_n$.
In particular,
the first few terms for $I_0$ given above also gives the first few terms
of $w$  with this change.
Similarly,
the Feynman rules that expresses $I_0$ as an enumeration
of the rooted trees also work for $w$ with this change.
So after we apply these ideas to \eqref{eqn:Thin-Lag},
we prove the following:

\begin{thm} \label{thm:z-in-v}
The thin deformation of the thin spectral curve of Hermitian $N \times N$-matrix model
can alternatively be given by a formal power series in  $v=y/\sqrt{2}$
defined explicitly by:
\be \label{eqn:Gen-Catalan}
z = I_0 +  \sum_{k=1}^\infty \frac{v^k}{k}
\sum_{p_1 + \cdots + p_k = k-1} \frac{J_{p_1}}{p_1!} \cdots
\frac{J_{p_k}}{p_k!},
\ee
where $J_n$ are given by
\begin{align}
J_0 & =2N, & J_1 &=0, &
J_{n+1} & = (n+1)(I_n - \delta_{n,1}), \quad n \geq 1.
\end{align}
Furthermore,
this series is also given by enumeration over rooted trees
as in \eqref{eqn:I0-Feynman-1} - \eqref{eqn:I0-Feynman-3}
but with $t_n$ changed to $J_n$.
\end{thm}

The following are the first few terms of $z$:
\ben
z & = & I_0 + (2N)v+\frac{1}{2}(2N)^2(2\tilde{I}_1)v^3
+\frac{1}{6}(2N)^3(3I_2)v^4 \\
& + & \biggl(\frac{1}{24}(4I_3)(2N)^4+\frac{1}{2}(2N)^3(2\tilde{I}_1)^2\biggr)v^5 \\
& + & \biggl(\frac{1}{120} (2N)^5(5I_4)
+\frac{5}{12} (2N)^4(2\tilde{I}_1)(3I_2) \biggr)v^6 \\
& + & \biggl(\frac{1}{720}(2N)^6(6I_5) + \frac{1}{12}(2N)^5(3I_2)^2 +
 \frac{1}{8}(2N)^5(2\tilde{I}_1)(4I_3)
+\frac{5}{8}(2N)^4(2\tilde{I}_1)^3\biggr)v^7 \\
& + & \biggl(\frac{1}{5040}(2N)^7(7I_6)+\frac{7}{240} (2N)^6(2\tilde{I}_1)(5I_4)
+\frac{7}{144}(4I_3)(2N)^6(3I_2)\\
&& +\frac{7}{8}  (2N)^5(2\tilde{I}_1)^2 (3I_2)\biggr)v^8 \\
& +& \biggl(\frac{1}{40320} (2N)^8(8I_7)
+ \frac{1}{144}(2N)^7(4I_3)^2 +\frac{1}{180}  (2N)^7(2\tilde{I}_1)(6I_5) \\
&& +  \frac{1}{90}(2N)^7(3I_2) (5I_4) + \frac{7}{8}(2N)^5(2\tilde{I}_1)^4
+ \frac{7}{24} (2N)^6(2\tilde{I}_1)^2(4I_3) \\
&& + \frac{7}{18} (2N)^6(3I_2)^2(2\tilde{I}_1) \biggr)v^9 +\cdots
\een

So far in this Subsection we have used the coordinates $t_n$ or $I_n$,
as a result,
the coefficients of various expression are often fractional numbers.
To obtain integral coefficients, now we change to the coordinates $g_n$.

\begin{prop} \label{prop:Ik}
When expressed as formal power series
in the coordinates $\{g_n\}_{n \geq 1}$,
the renormalized coupling constants $\frac{1}{k!}I_k$ have integral coefficients.
\end{prop}

\begin{proof}
By \eqref{eqn:I0},
\ben
I_0 & = & \sum_{k=1}^\infty \frac{1}{k}
\sum_{p_1 + \cdots + p_k = k-1} g_{p_1+1} \cdots g_{p_k+1} \\
& = &  \sum_{k=1}^\infty \frac{1}{k}
\sum_{\substack{\sum_{i=1}^k m_i= k
\\ \sum_{i=1}^k (i-1)m_i = k-1}}
\frac{(\sum_{i=1}^k m_i)!}{m_1!\cdots m_k!}
g_1^{m_1} \cdots g_k^{m_k} \\
& = & \sum_{\substack{\sum_{i=1}^k m_i= k
\\ \sum_{i=1}^k (i-1)m_i = k-1}} \frac{1}{\sum_{i=1}^k m_i}
\binom{\sum_{i=1}^k m_i}{m_1,\cdots, m_k}
g_1^{m_1} \cdots g_k^{m_k}.
\een
Now since $(k,k-1)=1$,
we know the greatest common divisor of $m_1, \dots, m_k$ is equal to $1$,
i.e.,
\be
(m_1, \dots, m_k) = 1
\ee
therefore,
by \cite[Propospition 3.1]{Zhou-Mirror},
\be
\frac{1}{\sum_{i=1}^k m_i}
\binom{\sum_{i=1}^k m_i}{m_1,\cdots, m_k} \in \bZ.
\ee
This shows that $I_0$ has integral coefficients.
Secondly, by \eqref{eqn:Ik}, for $k \geq 1$,
\ben
\frac{I_k}{k!} & = & \sum_{n \geq 0} \binom{n+k}{k}
g_{n+k} I_0^n,
\een
and so $\frac{I_k}{k!}$ also has integral coefficients.
\end{proof}

\begin{remark}
As the examples in earlier Subsections show,
the integral coefficients of $I_0$ as formal series in $\{g_n\}_{n\geq 1}$
can be thought of as various generalizations of the Catalan numbers.
\end{remark}

As a corollary of Proposition \ref{prop:Ik} and Theorem \ref{thm:z-in-v},
we then have:

\begin{thm} 
The thin deformation of the thin spectral curve of
Hermitian $N \times N$-matrix model
can be given by a formal power series in  $v=y/\sqrt{2}$, $2N$
and $\{g_n\}_{n \geq 1}$
with integral coefficients.
\end{thm}

It is clear that the thin special deformation
written in the form of \eqref{eqn:Gen-Catalan}
is a generalization of \eqref{eqn:Signed-Catalan},
a generating series of (signed) Catalan numbers.
Therefore, when it is written as formal power series
in $1/(\sqrt{2}y)$, $(2N)$ and $\{g_n\}_{n \geq 0}$,
the coefficients can again be regarded as generalizations
of Catalan numbers.
So we have reached the following surprising byproduct
as the thin emergent geometry of Hermitian matrix models:
The thin special deformations of the thin spectral curves
of Hermitian matrix models can be used
to define sophisticated generalizations of Catalan numbers.

\section{Fat Emergent Geometry of  Hermitian One-Matrix Models}

\label{sec:Fat}

In this Section we discuss the fat special deformation
of the fat spectral curve of Hermitian one-matrix models.
In the literature (see e.g. \cite[\S 3.2]{DiF}
or \cite[\S 2.2]{Marino}),
the discussions in this Section is usually carried out in the setting of
taking $N \to \infty$.
It is crucial for us to note that it also works for finite $N$.

\subsection{Fat special deformation}

Recall the fat Virasoro constraints are given by the following differential operators
(cf. \cite[\S 3.6]{Zhou-Fat-Thin}):
\ben
&& L_{-1,t} = - \frac{\pd}{\pd g_1}
+  \sum_{n \geq 1} ng_{n+1} \frac{\pd}{\pd g_n} +   tg_1g_s^{-2}, \\
&& L_{0,t} = - 2\frac{\pd}{\pd g_2} +  \sum_{n \geq 1} ng_n \frac{\pd}{\pd g_n}
+  t^2g_s^{-2}, \\
&& L_{1,t} = - 3\frac{\pd}{\pd g_3} + \sum_{n \geq 1} (n+1)g_n \frac{\pd}{\pd g_{n+1}}
+ 2t\frac{\pd}{\pd g_1}, \\
&& L_{m,t} = \sum_{k \geq 1} (k+m) (g_k-\delta_{k,2}) \frac{\pd}{\pd g_{k+m}}
+ g_s^2 \sum_{k=1}^{m-1} k(m-k)\frac{\pd}{\pd g_k} \frac{\pd}{\pd g_{m-k}}
+ 2 tm  \frac{\pd}{\pd g_m},
\een
where $m \geq 2$.
As in the thin case,

\begin{thm} \label{thm:Fat-Existence}
Consider the following series:
\be
y = \frac{1}{\sqrt{2}} \sum_{n=1}^\infty (g_n -\delta_{n,2}) z^{n-1}
+ \frac{\sqrt{2}t}{z} + \sqrt{2}
\sum_{n = 1}^\infty \frac{n}{z^{n+1}} \frac{\pd F_0(t)}{\pd g_n}.
\ee
Then one has:
\be \label{eqn:Fat-Special-Def}
(y^2)_- =  0.
\ee
\end{thm}

\begin{proof}
This is actually equivalent to the Virasoro constraints for $F_0(t)$.
Indeed,
\ben
\frac{y^2}{2} & = & \biggl( \half \sum_{n=1}^\infty (g_n -\delta_{n,2}) z^{n-1}
+ \frac{1}{z} +
\sum_{n = 1}^\infty \frac{n}{z^{n+1}} \frac{\pd F_0(t)}{\pd g_n} \biggr)^2  \\
& = & \frac{1}{4} \biggl( \sum_{n=1}^\infty (g_n -\delta_{n,2}) z^{n-1} \biggr)^2
+ \sum_{n=1}^\infty (g_n -\delta_{n,2}) z^{n-1}
\biggl( \frac{t}{z}
+  \sum_{n = 1}^\infty \frac{n}{z^{n+1}} \frac{\pd F_0(t)}{\pd g_n} \biggr) \\
& + & \biggl( \frac{t}{z}
+ \sum_{n = 1}^\infty \frac{n}{z^{n+1}} \frac{\pd F_0(t)}{\pd g_n} \biggr)^2.
\een
It follows that
\ben
\half (y^2)_- & = & \frac{1}{z} \biggl(t g_1
+ \sum_{n=1}^\infty n(g_{n+1} -\delta_{n,1}) \frac{\pd F_0(t)}{\pd g_{n}}  \biggr) \\
& + & \sum_{m \geq 0}\sum_{k=1}^\infty  \frac{(k+m)}{z^{m+2}} (g_k -\delta_{k,2})
\frac{\pd F_0(t)}{\pd g_{k+m}}   \\
& + & \biggl( \frac{t}{z}
+ \sum_{n = 1}^\infty \frac{n}{z^{n+1}} \frac{\pd F_0(t)}{\pd g_n} \biggr)^2.
\een
The fat Virasoro constraints in genus zero take the following form:
\ben
&& \frac{\pd F_0(t)}{\pd g_1} =
 \sum_{n \geq 1} ng_{n+1} \frac{\pd F_0(t)}{\pd g_n} + tg_1, \\
&& 2\frac{\pd F_0(t)}{\pd g_2} =  \sum_{n \geq 1} ng_n \frac{\pd F_0(t)}{\pd g_n}
+ t^2, \\
&& 3\frac{\pd  F_0(t)}{\pd g_3} = \sum_{n \geq 1} (n+1)g_n \frac{\pd F_0(t)}{\pd g_{n+1}}
+ 2t\frac{\pd F_0(t)}{\pd g_1}, \\
&&  (m+2) \frac{\pd  F_0(t)}{\pd g_{m+2}}
= \sum_{k \geq 1} (k+m) g_k \frac{\pd  F_0(t)}{\pd g_{k+m}}
+ \sum_{k=1}^{m-1} k(m-k)\frac{\pd  F_0(t)}{\pd g_k} \frac{\pd F_0(t)}{\pd g_{m-k}}
+ 2 tm  \frac{\pd F_0(t)}{\pd g_m},
\een
where $m \geq 2$.
By these, the proof is completed.
\end{proof}

We will refer to \eqref{eqn:Fat-Special-Def} as the {\em fat special deformation}
of the {\em fat spectral curve} to be discussed below.
Similar to the thin case, we can characterize
the fat special deformation as follows.

\begin{thm} \label{thm:Uniqueness-Fat}
There exists a unique series
\be
y = \frac{1}{\sqrt{2}} \sum_{n \geq 0} (v_n-\delta_{n,1}) z^{n} + \frac{\sqrt{2}t}{z}
+ \sqrt{2} \sum_{n \geq 0} w_n z^{-n-2}
\ee
such that
each $w_n \in \bC[[v_0, v_1, \dots]]$ and
\be \label{eqn:y2-Fat}
\half(y^2)_- =  0.
\ee
\end{thm}

\begin{thm}
For a series of the form
\be
y =  \frac{1}{\sqrt{2}} \frac{\pd S(z, \bt)}{\pd z}
+ \frac{\sqrt{2}N}{z}+ \sqrt{2} \sum_{n \geq 0} w_n z^{-n-2},
\ee
where $S$ is the universal action defined by:
\be \label{eqn:Action2}
S(x) = - \frac{1}{2}x^2 + \sum_{n \geq 1} t_{n-1} \frac{x^n}{n!},
\ee
and each $w_n \in \bC[[t_0, t_1, \dots]]$,
the equation
\be
\half (y^2)_- = 0
\ee
 has a unique solution given by:
\ben
&& y = \frac{1}{\sqrt{2}} \sum_{n=0}^\infty \frac{t_n -\delta_{n,1}}{n!} z^n
+ \frac{\sqrt{2}}{z} + \sqrt{2}
\sum_{n = 1}^\infty \frac{n!}{z^{n+1}} \frac{\pd F_{0}(t)}{\pd t_{n-1}}.
\een
\end{thm}

\subsection{Fat special deformation in terms of resolvent}

The resolvent $\omega$ is defined by:
\be \label{def:Resolvent}
\omega(z) = \frac{t}{z}
+\sum_{n=1}^\infty \frac{n}{z^{n+1}} \frac{\pd F_0(t)}{\pd g_n}.
\ee
For simplicity of notations,
write
\be
f_n:=n\frac{\pd F_0(t)}{\pd g_n}.
\ee
Then by the fat Virasoro constraints in genus zero we have:
\ben
\omega(z) & = & \frac{t}{z}
+ \frac{1}{z^2} \biggl(\sum_{n \geq 1} g_{n+1} f_n + tg_1\biggr) \\
& + & \frac{1}{z^3} \biggl( \sum_{n \geq 1} g_n f_n
+ t^2\biggr)
+ \frac{1}{z^4}\biggl( \sum_{n \geq 1} g_n f_{n+1}
+ 2tf_1 \biggr) \\
& + & \sum_{m \geq 2} \frac{1}{z^{m+2}}
\biggl(\sum_{k \geq 1} g_k f_{k+m}
+ \sum_{k=1}^{m-1} f_k f_{m-k}
+ 2 tm  f_m \biggr) \\
& = & \frac{1}{z} \omega^2(z) + \frac{t}{z}+ \frac{g_1}{z} \omega(z)
+ g_2 \biggl(\omega(z) - \frac{t}{z}\biggr) \\
& + & g_3z \biggl(\omega(z) - \frac{t}{z} - \frac{f_1}{z^2} \biggr)+ \cdots \\
& = & \frac{1}{z} \omega^2(z) +\frac{1}{z} (g_1 +g_2z+g_3z^2+\cdots) \omega(z) \\
& + & \biggl(\frac{t}{z} - g_2 \frac{t}{z} - g_3z\biggl(\frac{t}{z} + \frac{f_1}{z^2}\biggr)
- \cdots \biggr).
\een
It can be rewritten in the following form:
\ben
&& \omega^2(z) + (g_1+(g_2-1)z +g_3z^2+\cdots) \cdot \omega(z) \\
& -& (g_2-1) t  - g_3z^2\biggl(\frac{t}{z} + \frac{f_1}{z^2}\biggr)    - \cdots
= 0 \een
In other words,
the fat special deformation in term of the resolvent is given by:
\be \label{eqn:Spec-Resolvent}
\omega^2(z) +S'(z)\cdot \omega(z) +P(z) = 0£¬
\ee
where $S$ is the action
\be
S(x) := \sum_{n=1}^\infty \frac{1}{n} (g_n- \delta_{n,2}) x^n,
\ee
and $P(z)$ is determined by $S(z)$ as follows:
\be \label{eqn:P}
P(z) = -(g_2-1) t  - g_3z^2\biggl(\frac{t}{z} + \frac{f_1}{z^2}\biggr)    - \cdots.
\ee
It is clear from the above formulas
when one takes $g_k = 0 $ for $k \geq n$,
then to find the fat spectral curve it suffices
to find only $f_1, f_2, \dots, f_{n-2}$.
From \eqref{eqn:Spec-Resolvent} we get:
\be
\omega(z) = \frac{-S'(z)- \sqrt{S'(z)^2 - 4P(z)}}{2}.
\ee

For the treatment from the large $N$ point of view,
see e.g. \cite[\S 3.2]{DiF}.

\subsection{The emergence of semi-circle law}

Let us first let all $g_k=0$ for all $k \geq 0$.
Then we have
\begin{align*}
S(x)& = -\half x^2, & P(z) & = t,
\end{align*}
and so
\be \label{eqn:Catalan}
\omega(z) = \frac{z-\sqrt{z^2-4t}}{2}.
\ee
It is well-known that
\be \label{eqn:Fat-Catalan}
\omega(z) = \sum_{n \geq 0}C_n \frac{t^{n+1}}{z^{2n+1}},
\ee
where $C_n$ are the Catalan numbers:
\be
C_n = \frac{1}{n+1}\binom{2n}{n}.
\ee
This means that
\be
n \frac{\pd F_0(t)}{\pd g_n}\biggl|_{g_k =0, k \geq 1}
= \begin{cases}
C_m, & n = 2m, \\
0, & \text{otherwise}.
\end{cases}
\ee

From \eqref{eqn:Catalan} we get:
\be \label{eqn:Fat-Spectral}
\omega^2 - z \omega + t = 0,
\ee
or equivalently,
\be \label{eqn:Fat-Spectral2}
z = \omega +\frac{t}{\omega}.
\ee
By comparing with \eqref{eqn:Spec-Thin},
we see a strange duality.

On the other hand,
\be
y = \frac{-\sqrt{z^2-4t}}{\sqrt{2}}.
\ee
So the spectral curve in terms of the field $y=y(z)$ is given by the algebraic curve:
\be \label{eqn:Fat-Spec}
z^2 - 2y^2 = 4t.
\ee

For the treatment from the large $N$ point of view,
see e.g. the Example in \cite[\S 2.2]{Marino}.

We will refer to the plane algebraic curve \eqref{eqn:Fat-Spec} as
the {\em fat spectral curve} of Hermitian one-matrix models.
It is equivalently given in terms of the resolvent by \eqref{eqn:Catalan}
or \eqref{eqn:Fat-Catalan}.
These equations are the dual versions of \eqref{eqn:Signed-Catalan},
so it is tempting to say that the thin spectral curve and the fat spectral curve
are dual to each other,
furthermore, the thin special deformation defined by \eqref{eqn:Gen-Catalan}
and the fat deformation defined by \eqref{def:Resolvent} are dual to each other.
Because $F_0(t)$ are formal power series in $t$ and $\{g_n\}_{n \geq 1}$ with integral
coefficients,
we see that the fat special deformation also can be used to define generalizations
of Catalan numbers.
In the following several Subsections
we will examine some examples
to gain more insights into the nature of the plane algebraic curves
defined by the fat special deformations.

\subsection{Fat special deformation along the $g_1$-line}
\label{sec:Fat-g1}
Let us first let all $g_k=0$ for all $k \geq 1$.
Then we have
\begin{align*}
S(x)& = -\half x^2 +g_1x, & P(z) & = t,
\end{align*}
and so
\be \label{eqn:Motzkin}
\omega(z) = \frac{z-g_1-\sqrt{(z-g_1)^2-4t}}{2}.
\ee
Its expansion is related to the {\em Motzkin polynomials}:
\be
R_n(x): = \sum_{k=0}^{[n/2]} T(n, k) x^k,
\ee
where $T_{n,k}$ are defined by:
\be
T(n,k) = \frac{n!}{(n-2k)! k! (k+1)!}.
\ee
More precisely,
\be
\omega(z) = \sum_{n \geq 0} \sum_{k=0}^{[n/2]} T(n,k) \frac{t^{k+1}g_1^{n-2k}}{z^{n+1}},
\ee
This mean
\be
n \frac{\pd F_0(t)}{\pd g_n}\biggl|_{g_k =0, k \geq 2}
=  \sum_{k=0}^{[n/2]} T(n,k) t^{k+1}g_1^{n-2k}.
\ee

Alternatively,
write
\ben
f_n:=n\frac{\pd F_0(t)}{\pd g_n} \biggl|_{g_k =0, k \geq 2}.
\een
Then by the fat Virasoro constraints one has the following recursion relations:
\be
\begin{split}
&f_1 =   tg_1, \\
&f_2 =  g_1 f_1 +  t^2, \\
&f_3 = g_1 f_2+ 2tf_1, \\
& f_{m+2} = g_1 f_{m+1} + \sum_{k=1}^{m-1} f_k f_{m-k} + 2 t f_m,
\qquad m \geq 1.
\end{split}
\ee
It is natural to set
\be
f_0: = t.
\ee
Then one has
\ben
\omega(z) &= &\sum_{n \geq 0}\frac{f_n}{z^{n+1}}
= \frac{t}{z} + \frac{g_1}{z} \omega(z) + \frac{1}{z} \omega^2(z),
\een
and from this one can again get \eqref{eqn:Motzkin}.

The first few of $f_n$ are
\ben
f_0 &=& t, \\
f_1 &=& t g_1, \\
f_2 &=& t g_1^2+t^2, \\
f_3 &=& t g_1^3+3t^2g_1, \\
f_4 &=& t g_1^4+6t^2g_1^2+2t^3, \\
f_5 &=& t g_1^5+10t^2g_1^3+10t^3g_1, \\
f_6 &=& t g_1^6+15t^2g_1^4+30t^3g_1^2+5t^4, \\
f_7 &=& t g_1^7+21t^2g_1^5+70t^3g_1^3+35t^4g_1.
\een

From \eqref{eqn:Motzkin} we get:
\be
\omega^2 - (z-g_1) \omega + t = 0,
\ee
and
\be
z = \omega +\frac{t}{\omega} + g_1.
\ee
These are deformations of \eqref{eqn:Fat-Spectral2} and \eqref{eqn:Fat-Spectral} respectively.
The field $y$ is deformed to
\be \label{eqn:Fat-Def-g1}
y = \frac{-\sqrt{(z-g_1)^2-4t}}{\sqrt{2}},
\ee
and so \eqref{eqn:Fat-Spec} is deformed to£º
\be \label{eqn:Fat-Def-g1-2}
(z-g_1)^2 - 2y^2 = 4t.
\ee
So as in the thin case,
the fat special deformation along the $g_1$-line is obtained by changing $z$ to $z-g_1$.

\subsection{Fat special deformation   along the $g_2$-line}
\label{sec:Fat-g2}

Next we let $g_n=0$  for $n \neq 2$,
i.e.,
\begin{align*}
S(x) & = - \half (1-g_2) x^2, &
P(z) & = (1-g_2)t.
\end{align*}
From this we get the resolvent:
\be \label{eqn:Motzkin-g2}
\omega(z) = \frac{(1-g_2)z - \sqrt{(1-g_2)^2z^2 - 4(1-g_2)t}}{2},
\ee
and so
\be \label{eqn:Fat-Def-g2}
\frac{y}{\sqrt{2}} = \half (g_2-1)z+ \omega(z)= - \frac{\sqrt{(1-g_2)^2z^2-4(1-g_2)t}}{2}.
\ee
Therefore,
the fat special deformation along the $g_2$-line leads to the plane algebraic curve:
\ben
&& (1-g_2)^2 z^2 - 2 y^2 = 4 (1-g_2) t,
\een
this can be obtained from \eqref{eqn:Fat-Spec} by
making the following changes:
\begin{align*}
y & \mapsto \frac{y}{1-g_2}, & t &\mapsto \frac{t}{1-g_2}.
\end{align*}

The expansion of the resolvent in this case is
\ben
\omega(z)
& = & \sum_{n=0}^\infty \frac{1}{n+1} \binom{2n}{n} \frac{t^{n+1}}{(1-g_2)^{2n+1}}
 \frac{1}{z^{2n+1}}.
\een
From this one gets
\be
n\frac{\pd F_0}{\pd g_n} \biggl|_{g_k=0, k\neq 2}
= \frac{1}{n+1} \binom{2n}{n} \frac{t^{n+1}}{(1-g_2)^{2n+1}}.
\ee

\subsection{Fat special deformations on the $(g_1,g_2)$-plane}
\label{sec:Fat-g1g2}

Next we let $g_n=0$  for $n \geq 3$,
i.e.,
\begin{align*}
S(x) & = g_1x- \half (1-g_2) x^2, &
P(z) & = (1-g_2)t.
\end{align*}
From this we get the resolvent:
\be  \label{eqn:omega-g1g2}
\omega(z) = \frac{(1-g_2)z-g_1 - \sqrt{[(1-g_2)z-g_1]^2 - 4(1-g_2)t}}{2},
\ee
and so
\be
\frac{y}{\sqrt{2}} = \half ((g_2-1)z+g_1)+ \omega(z)
= - \frac{\sqrt{((1-g_2)z-g_1)^2-4(1-g_2)t}}{2}.
\ee
This can be obtained from \eqref{eqn:Fat-Def-g1}
 by
making the following changes:
\begin{align*}
y & \mapsto \frac{y}{1-g_2}, & t &\mapsto \frac{t}{1-g_2}, &
g_1 & \mapsto \frac{g_1}{1-g_2}.
\end{align*}
We rewrite \eqref{eqn:omega-g1g2} in the following form£º
\be
\bigg(\frac{\sqrt{2}\omega}{1-g_2}\biggr)^2
- \sqrt{2}\frac{\sqrt{2}\omega}{1-g_2} \biggl(z-\frac{g_1}{1-g_2}\biggr)
+ \frac{2t}{1-g_2} = 0.
\ee
This can be regarded as a dual to \eqref{eqn:Thin-Def-g1g2}:
$\frac{\sqrt{2}\omega}{1-g_2}$ in this equation plays the role of
$z-\frac{g_1}{1-g_2}$ in \eqref{eqn:Thin-Def-g1g2},
$z-\frac{g_1}{1-g_2}$ in this equation plays the file of
$\frac{y}{1-g_2}$ in \eqref{eqn:Thin-Def-g1g2},
and $\frac{2t}{1-g_2}$ in this equation plays the role of
$-\frac{2N}{1-g_2}$ in \eqref{eqn:Thin-Def-g1g2}.

The expansion of \eqref{eqn:omega-g1g2} is
\be
\omega(z) = \sum_{n \geq 0} \sum_{k=0}^{[n/2]} T(n,k)
\frac{t^{k+1}g_1^{n-2k}}{(1-g_2)^{n-k}z^{n+1}},
\ee
from this one can get closed formula for $\corr{p_1^k(\frac{p_2}{2})^lp_n}_0(t)$.

\subsection{Fat special deformations along the $g_3$-line}
\label{sec:Fat-g3}

Next we let $g_n=0$  for $n \neq 2$,
i.e.,
\begin{align*}
S(x) & = - \half  x^2 + \frac{g_3}{3} x^3, &
P(z) & = t  - g_3z^2\biggl(\frac{t}{z} + \frac{f_1}{z^2}\biggr),
\end{align*}
where
\ben
&& f_n = n\frac{\pd F_0(t)}{\pd g_n} \biggl|_{g_k =0, k \neq 3}.
\een
From this we get the resolvent:
\be \label{eqn:omega-g3}
\omega(z) = \frac{z -g_3z^2 -
\sqrt{(z-g_3z^2)^2  - 4(t  - g_3tz - g_3f_1)}}{2},
\ee
and so
\be
y = \half (g_3z^2-z)+ \omega(z)= - \frac{\sqrt{(z-g_3z^2)^2 - 4(t  - g_3tz - g_3f_1)}}{2}.
\ee
The spectral curve is then deformed to the following hyperelliptic curve on the
$(y,z)$-plane:
\be \label{eqn:SpecCurv-g3}
4y^2 = (z-g_3z^2)^2 - 4(t  - g_3tz - g_3f_1).
\ee
To make this result more explicit,
we need to find explicit formula for $f_1$.
Note in the above expressions,
\be
f_1 = \sum_{n \geq 0} \frac{1}{n!} \corr{p_1 \big(\frac{p_3}{3}\big)^n}_0(t) \cdot g_3^n,
\ee
where $\corr{p_1 \big(\frac{p_3}{3}\big)^n}_0(t)$ denote the fat genus zero correlators
in the sense of \cite{Zhou-Fat-Thin}.
For small $n$,
these correlators can be computed using the fat Virasoro constraints,
\ben
f_1 & = & \sum_{n \geq 1} \frac{1}{n!} \cdot n \cdot
\corr{p_2 \big(\frac{p_3}{3}\big)^{n-1}}_0(t) \cdot g_3^n \\
& = & g_3t^2 +\sum_{n \geq 2} \frac{1}{(n-1)!} \cdot 3 (n-1) \cdot
\corr{\big(\frac{p_3}{3}\big)^{n-1}}_0(t) \cdot g_3^n \\
& = & g_3 t^2 + 3 \sum_{m\geq 1} \frac{1}{(2m-1)!} \cdot
\corr{\big(\frac{p_3}{3}\big)^{2m}}_0(t) \cdot g_3^{2m+1}.
\een
Using the fat graphs,
one see that the computation of
$\frac{1}{n!} \corr{p_1 \big(\frac{p_3}{3}\big)^n}_0(t)$
is equivalent to the problem of rooted planar trivalent graphs,
hence is given by the sequence A002005 on \cite{Sloane},
and so by e.g. \cite[\S 3]{Krasko-Omel} we have for $n \geq 1$,
\be
\frac{1}{n!} \corr{p_1 \big(\frac{p_3}{3}\big)^{n}}_0(t)
= \begin{cases}
\frac{2^{2m+1}(3m)!!}{(m+2)!m!!} t^{m+2}, & n = 2m+1, \\
0, & n=2m.
\end{cases}
\ee
One can derive from this
\be
\corr{p_1 \big(\frac{p_3}{3}\big)^{2m+1}}_0(t)
= \frac{2^{2m+1}(2m+1)!(3m)!!}{(m+2)!m!!} t^{m+2},
\ee
and
\be
\corr{\big(\frac{p_3}{3}\big)^{2m}}_0(t)
= \frac{2^{2m+1}(2m-1)!(2m+1)!(3m)!!}{3\cdot (m+2)!m!!} t^{m+2}.
\ee
It follows that
\be \label{eqn:f1-g3}
f_1 = \sum_{m\geq 0} \frac{2^{2m+1}(3m)!!}{(m+2)!m!!} t^{m+2} g_3^{2m+1},
\ee
and so the specially deformed spectral curve
 \eqref{eqn:SpecCurv-g3}  along the $g_3$-line is explicitly given by
\be
y^2 = \frac{1}{4}(z-g_3z^2)^2 - t  + g_3tz
+ \sum_{m\geq 0} \frac{2^{2m+1}(3m)!!}{(m+2)!m!!} t^{m+2} g_3^{2m+2}.
\ee
To understand this hyperelliptic curve,
we need to study the discriminant $\Delta$ of the right-hand side.
Using Maple,
we find:
\be
\Delta=\frac{1}{16}g_3^5\biggl(64g_3^4f_1^3+(g_3-96g_3^3t)f_1^2
+(30g_3^2t^2-t)f_1-27g_3^3t^4+g_3t^3\biggr),
\ee
this vanishes because it is known from Entry  A002005 on \cite{Sloane}
that if
\ben
y(x) = \sum_{m\geq 0} \frac{2^{2m+1}(3m)!!}{(m+2)!m!!} x^m,
\een
then
\be
64x^3y^3 + x(1-96x)y^2 + (30x-1)y - 27x + 1 = 0.
\ee
It follows that the specially deformed spectral curve has the following form:
\be
y^2 = (1-g_3z+\alpha)^2(z-a_+)(z-a_-).
\ee
In \S \ref{sec:Fat-g3-again} we will go in the reverse direction and use this as a method
to compute computing $f_1$ and hence $\corr{\big(\frac{p_3}{3}\big)^{2m}}_0(t)$
for all $m$.

The first few of terms of $f_1$ are:
\be
f_1 = g_3t^2 +4 g_3^3t^3 + 32g_3^5t^4 + 336g_3^7t^5 +4096g_3^9t^6
+ 54912 g_3^{11}t^7 + \cdots.
\ee
By \eqref{eqn:omega-g3} the first few terms of the resolvent are:
\ben
\omega & = & \frac{t}{z}+\frac{f_1}{z^2}+ \frac{f_1}{g_3z^3}
+\frac{f_1-g_3t^2}{g_3^2z^4}
+\frac{f_1-g_3t^2-2g_3^2tf_1}{g_3^3z^5} \\
& + & \frac{f_1-g_3t^2-4g_3^2tf_1-g_3^3f_1^2}{g_3^4z^6}
+\frac{f_1+2g_3^3t^3-g_3t^2-6g_3^2tf_1-3g_3^3f_1^2}{g_3^5z^7}+\cdots \\
& = & \frac{t}{z}
+ \frac{1}{z^2} (g_3t^2 +4 g_3^3t^3 + 32g_3^5t^4 + 336g_3^7t^5
+4096g_3^9t^6 + \cdots) \\
& + & \frac{1}{z^3} (t^2+4g_3^2t^3+32g_3^4t^4+336g_3^6t^5
+4096g_3^8t^6 + \cdots)\\
& + & \frac{1}{z^4} (4g_3t^3+32g_3^3t^4+336g_3^5t^5+4096g^7t^6+
54912g_3^9t^7 +\cdots) \\
& + & \frac{1}{z^5} (2t^3+24g_3^2t^4+272g_3^4t^5+3424g_3^6t^6+46720g_3^8t^7
+ \cdots) \\
& + & \frac{1}{z^6} (15g_3t^4+200g_3^3t^5+2672g_3^5t^6+37600g_3^7t^7
+ \cdots) + \cdots.
\een
We regard this as a dual to \eqref{eqn:z-g3},
in both cases the right-hands
are formal power series with integral coefficients which generalize
the Catalan numbers.
The leading terms of the coefficients of $\frac{1}{z^{2n+1}}t^{n+1}$
are $\frac{1}{n+1}\binom{2n}{n}$.
We have checked that the leading terms of the coefficients of $\frac{1}{z^{2n}}$
are $g_3t^{n+1}$ times $\binom{2n}{n-1}$ which are
the sequence A001791 on \cite{Sloane}.

From the above calculation of $\omega$,
one also gets a way to compute the first  few $f_n$'s
and hence the corresponding correlators:
\ben
f_2 & = & \frac{f_1}{g_3}
= \sum_{m\geq 0} \frac{2^{2m+1}(3m)!!}{(m+2)!m!!} t^{m+2} g_3^{2m},\\
f_3 & = &\frac{f_1-g_3t^2}{g_3^2}
= \sum_{m\geq 1} \frac{2^{2m+1}(3m)!!}{(m+2)!m!!} t^{m+2} g_3^{2m-1}, \\
f_4 & = & \frac{f_1-g_3t^2-2g_3^2tf_1}{g_3^3} \\
& = & \sum_{m\geq 1} \frac{2^{2m+1}(3m)!!}{(m+2)!m!!} t^{m+2} g_3^{2m-2}
- 2 \sum_{m\geq 0} \frac{2^{2m+1}(3m)!!}{(m+2)!m!!} t^{m+3} g_3^{2m}\\
& = & \sum_{m \geq 0}\biggl( \frac{2^{2m+3}(3m+3)!!}{(m+3)!(m+1)!!}
- \frac{2^{2m+2}(3m)!!}{(m+2)!m!!}\biggr) t^{m+3} g_3^{2m}, \\
f_5 & = & \frac{f_1-g_3t^2-4g_3^2tf_1-g_3^3f_1^2}{g_3^4}, \\
f_6 & = & \frac{f_1+2g_3^3t^3-g_3t^2-6g_3^2tf_1-3g_3^3f_1^2}{g_3^5}.
\een
The coefficients of $f_4$ are the integer sequence A002006
on \cite{Sloane}: $2$, $24$, $272$, $3424$, $46720$, $\cdots$;
the coefficients of $f_5$ are the integer sequence A002007:
$15$,  $200$, $2672$, $37600$, $\cdots$;
the coefficients of $f_6$ are the integer sequence A002008£º
$5$, $120$, $1840$, $27552$, $\cdots$;
the coefficients of $f_7$ and $f_8$ are the sequences A002009 and A002010 respectively.
These numbers are the numbers of almost trivalent maps in the sense of \cite{Mullin-etal}.
Our discussion in this Subsection
shows that the plane algebraic curve \eqref{eqn:SpecCurv-g3} encode
all the numbers of almost trivalent maps.

\subsection{Fat special deformations on the $(g_1, g_2,g_3)$-space}
\label{sec:Fat-g1g2g3}

Next we let $g_n=0$  for $n \geq 4$,
i.e.,
\begin{align*}
S(x) & = - \half  x^2 + g_1x +\frac{g_2}{2}x^2 + \frac{g_3}{3} x^3, &
P(z) & = (1-g_2) t  - g_3z^2\biggl(\frac{t}{z} + \frac{f_1}{z^2}\biggr),
\end{align*}
where
\ben
&& f_1 = \frac{\pd F_0(t)}{\pd g_1} \biggl|_{g_k =0, k \geq 4}.
\een
From this we get the resolvent:
\be \label{eqn:omega-g1g2g3}
\begin{split}
 \omega(z)  = & \frac{1}{2} \biggl(-g_1+(1-g_2)z-g_3z^2 \\
 - & \sqrt{ (-g_1+(1-g_2)z-g_3z^2)^2- 4((1-g_2)t  - g_3tz - g_3f_1)} \biggr),
\end{split}
\ee
and so
\be
y =   - \frac{
\sqrt{ (-g_1+(1-g_2)z-g_3z^2)^2- 4((1-g_2)t  - g_3tz - g_3f_1)}}{2}.
\ee
The spectral curve is then deformed to
\be \label{eqn:Spec-g1g2g3}
4y^2 = (-g_1+(1-g_2)z-g_3z^2)^2- 4((1-g_2)t  - g_3tz - g_3f_1).
\ee
Again, to make this special deformation more explicit,
we need to find explicit formula for $f_1$.
One can apply the fat dilaton equation and the fat string equation
to reduce the computation of $f_1$ to the computations of
$\corr{\big(\frac{p_3}{3}\big)^{2m}}_0(t)$.
The following are the first few terms of $f_1$:
\ben
f_1 & = & g_1t  +g_3t^2  + t g_2g_1
+ 2t^2 g_3g_2 + tg_3g_1^2 + t g_2^2g_1 + \cdots,
\een
On the other hand,
using Maple we find the discriminant of the degree four polynomial on the
right-hand side of \eqref{eqn:Spec-g1g2g3} is
\ben
\Delta & = & \frac{g_3^5}{16} \biggl(64g_3^4f_1^3
+g_3\big[\big((1-g_2)^2-4g_1g_3\big)^2-96(1-g_2)g_3^2t\big]f_1^2 \\
& + & \bigl[30tg_3^2(1-g_2)^2-(1-g_2)[(1-g_2)^2-4g_1g_3]^2+72tg_1g_3^3 \bigr] f_1 \\
& + & \biggl[-27t^4g_3^3+16t^2g_1^3g_3^2 -t^3g_3(1-g_2)^3
-t^2g_1(1-g_2)^4 \\
&& +8t^2g_1^2g_3(1-g_2)^2+36t^3g_1g_3^2(1-g_2) \biggr].
\een
One can use the vanishing of $\Delta$ to compute
$f_1$ and $\omega$.
This will be justified  in \S \ref{sec:Fat-g1g2g3-again}.

\subsection{Fat special deformations on the $g_4$-line}
\label{sec:Fat-g4}
Let $g_n=0$  for $n \neq 4$,
i.e.,
\begin{align*}
S(x) & = - \half  x^2 +\frac{g_4}{4}x^4, &
P(z) & = t  - g_4z^3\biggl(\frac{t}{z} + \frac{f_1}{z^2} + \frac{f_2}{z^3}\biggr),
\end{align*}
where
\begin{align*}
f_1 & = \frac{\pd F_0(t)}{\pd g_1} \biggl|_{g_k =0, k \neq 4}, &
f_2 & = 2\frac{\pd F_0(t)}{\pd g_2} \biggl|_{g_k =0, k \neq 4}.
\end{align*}
From this we get the resolvent:
\be \label{eqn:omega-g4}
\begin{split}
 \omega(z)  =  \frac{1}{2} \biggl(z-g_4z^3
 - \sqrt{ (z-g_4z^3)^2- 4(t  - g_4tz^2 - g_4f_1z - g_4f_2)} \biggr),
\end{split}
\ee
and so
\be
y =   - \frac{
\sqrt{ (z-g_4z^3)^2- 4(t  - g_4tz^2 - g_4f_1z-g_4f_2)}}{2}.
\ee
The special deformation of the spectral curve
in this case is
\be \label{eqn:Spec-g4}
4y^2 = (z-g_4z^3)^2- 4(t  - g_4tz^2 - g_4f_1z-g_4f_2).
\ee
We need to find explicit formulas for $f_1$ and $f_2$.
Note
\ben
f_1 = \sum_{n \geq 0} \frac{g_4^n}{n!}
\corr{p_1 \big(\frac{p_4}{4}\big)^n}_0^c(t) =0,
\een
where in the second equality we have used the fat
selection rule \cite[(21)]{Zhou-Fat-Thin}.
By the fat dilaton equation,
\ben
f_2 = \sum_{n \geq 0} \frac{g_4^n}{n!}
\corr{p_2 \big(\frac{p_4}{4}\big)^n}_0^c(t)
=  t^2+ \sum_{n \geq 1} 4n \frac{g_4^n}{n!}
\corr{\big(\frac{p_4}{4}\big)^n}_0^c(t).
\een
By using the fat graphs,
up to powers of $t$,
the correlators $\frac{1}{n!}
\corr{p_2 \big(\frac{p_4}{4}\big)^n}_0^c(t)$
are the number of rooted $4$-regular planar maps,
i.e.,
they are A000168 in \cite{Sloane}:
\be
\frac{1}{n!}
\corr{p_2 \big(\frac{p_4}{4}\big)^n}_0^c(t)
=  2\cdot 3^n\cdot \frac{(2n)!}{n!(n+2)!} t^{n+2}.
\ee
The first few terms of $f_2$ are:
\ben
f_2 & = & t^2  +2 t^3 g_4
+ 9t^4g_4^2+ 54t^5g_4^3+ 378t^6g_4^4+ 2916 t^7g_4^5
+ 24057t^8g_4^6+ \cdots.
\een
The generating series
\ben
&& A(z) = \sum_{n \geq 0} 2\cdot 3^n\cdot
\frac{(2n)!}{n!(n+2)!} z^n
\een
can be summed up:
$$A(z) = \frac{(1-12z)^{3/2}-1+18z}{54z^2},$$
hence $A$ satisfies the equation:
$$ 1 - 16z +( 18z-1)A(z) -27z^2A(z)^2 = 0.$$
The discriminant of the right-hand side of \eqref{eqn:Spec-g4}
is up to a constant
\be
(f_2g_4-t)g^{10}(27f_2^2g_4^2+16g_4t^3-18f_2g_4t-t^2+f_2)^2,
\ee
so it vanishes.
Conversely,
one can use this vanishing to compute $f_2$:
\be
f_2 =\frac{1}{54g_4^2} \biggl((1-12g_4t)^{3/2}-1+18g_4t\biggr).
\ee
For a different approach,
see \S \ref{sec:Fat-g4-again}.

The first few terms of the resolvent are
\ben
\omega & = & tx+\frac{f_2}{z^3}+\frac{f_2-t^2}{g_4z^5}
+\frac{f_2-t^2-2tg_4f_2}{g_4^2z^7} \\
& + & \frac{f_2+2g_4t^3-t^2-4tg_4f_2-g_4^2f_2^2 }{g_4^3z^9} \\
& + & \frac{f_2+4g_4t^3-t^2-6tg_4f_2-3g_4^2f_2^2+6g_4^2t^2f_2}{g_4^4z^{11}}+
\cdots,
\een
it follows that $f_{2n-1} = 0$ for $n \geq 1$,
and
\ben
f_4 & = & \frac{f_2-t^2}{g_4}
= \frac{1}{54g_4^3} \biggl((1-12g_4t)^{3/2}-1+18g_4t-54g_4^2t^2\biggr)  \\
& = &  2 t^3 + 9t^4g_4+ 54t^5g_4^2+ 378t^6g_4^3+ 2916 t^7g_4^4
+ 24057t^8g_4^5+ \cdots, \\
f_6 & = & \frac{f_2-t^2-2tg_4f_2}{g_4^2} \\
& = & \frac{1}{54g_4^3} \biggl((1-2g_4t)(1-12g_4t)^{3/2}-1+20g_4t-90g_4^2t^2\biggr) \\
& = & 5t^4 +36g_4t^5+270g_4^2t^6+2160g_4^3t^7 +18225g_4^4t^8 +160380g_4^5t^9  + \cdots, \\
f_8 & = & \frac{f_2+2g_4t^3-t^2-4tg_4f_2-g_4^2f_2^2 }{g_4^3} \\
& = & \frac{7}{54g_4^5} \biggl((2-9g_4t)(1-12g_4t)^{3/2}-2+45g_4t
-270g_4^2t^2+270g_4^3t^3\biggr) \\
& = & 14t^5+140g_4t^6+1260g_4^2t^7+11340g_4^3t^8 +103950g_4^4t^9 + \cdots,  \\
f_{10} & = & \frac{f_2+4g_4t^3-t^2-6tg_4f_2-3g_4^2f_2^2+6g_4^2t^2f_2}{g_4^4} \\
&= & 42t^6+540g_4t^7+5670g+4^2t^8+56700g_4^3t^9+561330g_4^4t^{10}+\cdots.
\een
The coefficients count fat graphs whose vertices are all $4$-valent except
for one vertex.
These numbers have not yet appeared on \cite{Sloane}.

\subsection{Fat special deformation along the $g_6$-line}
\label{sec:Fat-g6}
Let $g_n=0$  for $n \neq 6$,
i.e.,
\begin{align*}
S(x) & = - \half  x^2 +\frac{g_6}{6}x^6, &
P(z) & = t  - g_6z^5
\biggl(\frac{t}{z} + \frac{f_1}{z^2} + \frac{f_2}{z^3}
+ \frac{f_3}{z^4} + \frac{f_4}{z^5} \biggr),
\end{align*}
where
\begin{align*}
f_j & = j\frac{\pd F_0(t)}{\pd g_j} \biggl|_{g_k =0, k \neq 6}.
\end{align*}
From this we get the resolvent:
\be \label{eqn:omega-g6}
\begin{split}
 \omega(z)  =  \frac{1}{2} \biggl(z-g_6z^5
 - \sqrt{ (z-g_6z^5)^2- 4(t  - g_6tz^4
- g_6\sum_{j=1}^4 f_jz^{4-j} )} \biggr),
\end{split}
\ee
and so the special deformation of the spectral curve
in this case is
\be \label{eqn:Spec-g6}
4y^2 = (z-g_6z^5)^2- 4(t  - g_6tz^4 - g_6f_1z^3-g_6f_2z^2
- g_6f_3z-g_6f_4).
\ee
Note by the fat
selection rule \cite[(21)]{Zhou-Fat-Thin},
\ben
f_1 & = &\sum_{n \geq 0} \frac{g_6^n}{n!}
\corr{p_1 \big(\frac{p_6}{6}\big)^n}_0^c(t) =0, \\
f_3 & = &\sum_{n \geq 0} \frac{g_6^n}{n!}
\corr{\frac{p_3}{3} \big(\frac{p_6}{6}\big)^n}_0^c(t) =0,
\een
and by the fat dilaton equation,
\be \label{eqn:f2-g6}
f_2 = \sum_{n \geq 0} \frac{g_6^n}{n!}
\corr{p_2 \big(\frac{p_6}{6}\big)^n}_0^c(t)
=  t^2+ \sum_{n \geq 1} 6n \frac{g_6^n}{n!}
\corr{\big(\frac{p_6}{6}\big)^n}_0^c(t).
\ee
The first few terms of $f_2$ can be computed by e.g. \cite{Zhou-Mat-Mod}:
\ben
f_2 = t^2 + 5t^4g_6 + 100 t^6g_6^2 +\cdots
\een
Similarly, the first few terms of
\be \label{eqn:f4-g6}
f_4 = \sum_{n \geq 0} \frac{g_6^n}{n!}
\corr{p_4 \big(\frac{p_6}{6}\big)^n}_0^c(t)
\ee
can be computed:
\ben
f_4 = 2t^2 + 24t^4g_6  +\cdots.
\een
One can also consider the discriminant in this case,
but unfortunately its vanishing gives an equation that involves
both $f_2$ and $f_4$,
so one needs extra conditions.
In next Section we will find closed formulas for $f_2$ and $f_4$ in this case
by first establishing a result that is stronger than the vanishing of the discriminant.

\section{What is Really Special about Fat Special Deformations?}

\label{sec:Special}

The examples in last Section lead to a remarkable property of the plane algebraic
curves defined by fat special deformation.
Because of the lack of a terminology in the literature
we will say the fat special deformation is {\em formally one-cut}.
This is inspired by the method of one-cut solutions of Hermitian one-matrix models
in the literature.
Again this is often carried in the setting of large $N$-limit in the literature
(see e.g. \cite[\S 2.2]{Marino}),
but here we are working in the case with finite $N$.

\subsection{Fat special deformation is formally one-cut}

By this we mean  $S'(z)^2 - 4P(z)$ has the following form:
\be
S'(z)^2 - 4P(z) = Q(z)^2(z-a_-)(z-a_+),
\ee
where $Q(z)$ is a formal deformation of $1$,
$a_\pm$ are formal deformations of $\pm 2\sqrt{t}$ respectively:
\be
a_\pm = \pm 2\sqrt{t} + 2b_\pm(g_1, g_2, \dots).
\ee
I.e.,
$b_\pm$ are formal power series in $g_1, g_2, \dots$,
with coefficients in  $\bC[t^{1/2}]$,
such that
\be
b_\pm(g_1, g_2, \dots)|_{g_k=0, k \geq 1} = 0.
\ee
In other words,
we may write $\omega$ as
\be
\omega(z) = \frac{1}{2}
(-S'(z) - Q(z)\sqrt{(z - a_-(t))(z - a_+(t))}),
\ee

\begin{thm} \label{thm:Main}
The fat special deformation is formally one cut,
i.e., it has the form:
\be
y^2 = Q(z)^2(z-a_-)(z-a_+).
\ee
\end{thm}

To prove this Theorem, we will
present an algorithm that enables us to find $a_\pm$ and $Q$.
We adapt from a well-known method in the literature often used in the large $N$-limit.
Introduce
\be
H(z) = -\frac{S'(z)}{\sqrt{(z - a_-)(z - a_+)}},
\ee
considered as a series expansion for large $z$.
It is an element in $\bC[z, z^{-1}]]$£º
\be
H(z) = \sum_{i \in \bZ} H_{i} z^i.
\ee
Write $H(z) = H_{+}(z) + H_{-}(z)$,
where
\begin{align*}
H_{+}(z) & = \sum_{i \geq 0} H_{i} z^i, &
H_{-}(z) & = \sum_{i <0} H_{i}z^{-i}.
\end{align*}
Then one has $Q(z) = H_{+}(z)$ and
\be
\omega(z) = \frac{1}{2} H_{-}(z)\sqrt{(z - a_-)(z - a_+)}.
\ee
In particular, when we take $g_n = 0$ for $n \geq N$ for
some $N$,
i.e., $S(x)$ is just a polynomials,
then  $H_{i} = 0$ for $i \gg 0$.
$H_{+}(z)$ is the polynomial part of $H(z)$ in $z$.
Now
\be
\sqrt{(z-a_-)(z-a_+)} = z - \frac{a_-+a_+}{2} - \frac{(a_+-a_-)^2}{8z} + \cdots.
\ee
So one has
\be
\omega \sim \frac{H_{-1}}{2} +\frac{2H_{-2} - H_{-1}\cdot(a_++a_-)}{4z}+\cdots,
\ee
and by \eqref{def:Resolvent} we must have
\begin{align}
H_{-1}  &= 0, & H_{-2} &= 2t.
\end{align}
These coefficients can be easily computed.
One first gets:
\ben
H(z) & = & -\frac{S'(z)}{\sqrt{(z-a_+)(z-a_-)}} \\
& = & -\frac{S'(z)}{z} (1- \frac{a_+}{z})^{-1/2} \cdot (1- \frac{a_-}{z})^{-1/2} \\
& = & -\frac{S'(z)}{z} \sum_{i=0}^\infty (-1)^i \binom{-1/2}{i} \frac{a_+^i}{z^i}
\cdot \sum_{j=0}^\infty (-1)^j \binom{-1/2}{j} \frac{a_-^j}{z^j} \\
& = & (z - \sum^d_{m=1} g_mz^{m-1})\cdot
\frac{1}{z}\sum_{n=0}^\infty \sum_{i+j=n}
\frac{(2i-1)!!(2j-1)!!}{2^ni!j!} \frac{a_+^ia_-^j}{z^n},
\een
from this we get
\bea \label{eqn:H-1}
&&H_{-1} = - \sum_{n=0}^\infty \tilde{g}_{n+1} c_n = 0,\\
&&H_{-2} = - \sum_{n=1}^\infty \tilde{g}_n c_n = 2t, \label{eqn:H-2}
\eea
where $\tilde{g}_n = g_n - \delta_{n,2}$,
and $c_n$ is defined by:
\be
c_n:=\sum_{i+j=n}
\frac{(2i-1)!!(2j-1)!!}{2^ni!j!} a_+^ia_-^j
= \frac{1}{2^{2n}}\sum_{i+j=n} \binom{2i}{i} \binom{2j}{j} a_+^ia_-^j.
\ee
Change these into expressions in $b_+$ and $b_-$:
\ben
c_n& = & \frac{1}{2^{2n}} \sum_{i+j=n}
\binom{2i}{i} \binom{2j}{j} (2\sqrt{t}+2b_+)^i(-2\sqrt{t}-2b_-)^j \\
& = & \frac{1}{2^n}\sum_{i+j=n} (-1)^j
\binom{2i}{i} \binom{2j}{j}  \sum_{k=0}^i\binom{i}{k}t^{(i-k)/2}b_+^k
\sum_{l=0}^j\binom{j}{l}t^{(j-l)/2}b_-^l \\
& = & \frac{1}{2^n}\sum_{k,l} b_+^k b_-^l \sum_{i \geq k, j \geq l, i+j=n}
(-1)^j
\binom{2i}{i} \binom{2j}{j} \binom{i}{k} \binom{j}{l} t^{(n-k-l)/2}.
\een
The leading term of $c_n$ is
\ben
&& \frac{1}{2^n}\sum_{i+j=n}
(-1)^j \binom{2i}{i} \binom{2j}{j} \binom{i}{0} \binom{j}{0}  t^{n/2}
= \begin{cases}
\binom{2m}{m}t^m, & n = 2m, \\
0, & n = 2m-1.
\end{cases}
\een
This can be proved as follows:
\ben
&& \sum_{n=0}^\infty x^n \frac{1}{2^n}\sum_{i+j=n}
(-1)^j \binom{2i}{i} \binom{2j}{j}
= \sum_{i=0}^\infty \frac{1}{2^i}\binom{2i}{i}x^i \cdot
\sum_{j=0}^\infty (-1)^j \frac{1}{2^j} \binom{2j}{j}x^j \\
& = & (1-2x)^{-1/2} (1+2x)^{-1/2}
= (1-4x^2)^{-1/2}
= \sum_{m=0}^\infty \binom{2m}{m} x^{2m}.
\een
Similarly,
the subleading term is
\ben
&& b_+ \cdot \frac{1}{2^n}   \sum_{i \geq 1, j \geq 0, i+j=n}
(-1)^j
\binom{2i}{i} \binom{2j}{j} \binom{i}{1} \binom{j}{0}   t^{(n-1-0)/2}\\
& + & b_- \cdot \frac{1}{2^n}  \sum_{i \geq 0, j \geq 1, i+j=n}
(-1)^j
\binom{2i}{i} \binom{2j}{j} \binom{i}{0} \binom{j}{1}  t^{(n-0-1)/2} \\
& = & n\binom{n}{[n/2]}(b_++(-1)^n b_-) t^{(n-1)/2}.
\een
We have used the following two identities:
\ben
&& \frac{1}{2^n}   \sum_{i \geq 1, j \geq 0, i+j=n}
(-1)^j
\binom{2i}{i} \binom{2j}{j} \binom{i}{1} \binom{j}{0}
= n \binom{n-1}{[(n-1)/2]}, \\
&& \frac{1}{2^n}  \sum_{i \geq 0, j \geq 1, i+j=n}
(-1)^j
\binom{2i}{i} \binom{2j}{j} \binom{i}{0} \binom{j}{1}  t^{(n-0-1)/2}
= (-1)^n n \binom{n-1}{[(n-1)/2]}.
\een
The can be proved using genrating series as follows.
For the first identity, we have
\ben
&& \sum_{n=0}^\infty x^n \frac{1}{2^n}\sum_{i+j=n}
(-1)^j i \binom{2i}{i} \binom{2j}{j}
= \sum_{i=0}^\infty \frac{i}{2^i}\binom{2i}{i}x^i \cdot
\sum_{j=0}^\infty (-1)^j \frac{1}{2^j} \binom{2j}{j}x^j \\
& = & x \frac{d}{dx}\sum_{i=0}^\infty \frac{1}{2^i}\binom{2i}{i}x^i \cdot
\sum_{j=0}^\infty (-1)^j \frac{1}{2^j} \binom{2j}{j}x^j
= x \frac{d}{dx} (1-2x)^{-1/2} \cdot (1+2x)^{-1/2} \\
& =  & x(1-2x)^{-3/2} \cdot (1+2x)^{-1/2}
= x(1-2x)^{-1} \cdot (1-4x^2)^{-1/2},
\een
and on the other hand,
\ben
&& \sum_{n=1}^\infty n \binom{n-1}{[(n-1)/2]} x^n
= \sum_{m=1}^\infty 2m \binom{2m-1}{m-1} x^{2m}
+ \sum_{m=0}^\infty (2m+1) \binom{2m}{m} x^{2m+1} \\
& = & x\frac{d}{dx} \sum_{m=1}^\infty  \binom{2m-1}{m-1} x^{2m}
+ x\frac{d}{dx} \sum_{m=0}^\infty \binom{2m}{m} x^{2m+1} \\
& = & x\frac{d}{dx}\biggl(\frac{1-\sqrt{1-4x^2}}{2\sqrt{1-4x^2}}
+  \frac{x}{\sqrt{1-4x^2}} \biggr)
= x \frac{d}{dx} \frac{1+2x}{2(1-4x^2)^{1/2}}\\
& = & \frac{x}{(1-2x)(1-4x^2)^{3/2}},
\een
So the first identity is proved.
Similarly,
\ben
&& \sum_{n=0}^\infty x^n \frac{1}{2^n}\sum_{i+j=n}
(-1)^j j \binom{2i}{i} \binom{2j}{j}
= \sum_{i=0}^\infty \frac{1}{2^i}\binom{2i}{i}x^i \cdot
\sum_{j=0}^\infty (-1)^j \frac{j}{2^j} \binom{2j}{j}x^j \\
& = &  \sum_{i=0}^\infty \frac{1}{2^i}\binom{2i}{i}x^i \cdot
x \frac{d}{dx}\sum_{j=0}^\infty (-1)^j \frac{1}{2^j} \binom{2j}{j}x^j
= (1-2x)^{-1/2} \cdot x \frac{d}{dx} (1+2x)^{-1/2} \\
& =  & -(1-2x)^{-1/2} \cdot x(1+2x)^{-3/2}
= - x(1+2x)^{-1} \cdot (1-4x^2)^{-1/2},
\een
and
\ben
&& \sum_{n=1}^\infty (-1)^n n \binom{n-1}{[(n-1)/2]} x^n
= \sum_{m=1}^\infty 2m \binom{2m-1}{m-1} x^{2m}
- \sum_{m=0}^\infty (2m+1) \binom{2m}{m} x^{2m+1} \\
& = & x\frac{d}{dx} \sum_{m=1}^\infty  \binom{2m-1}{m-1} x^{2m}
- x\frac{d}{dx} \sum_{m=0}^\infty \binom{2m}{m} x^{2m+1} \\
& = & x\frac{d}{dx}\biggl(\frac{1-\sqrt{1-4x^2}}{2\sqrt{1-4x^2}}
-  \frac{x}{\sqrt{1-4x^2}} \biggr)
= x \frac{d}{dx} \frac{1-2x}{2(1-4x^2)^{1/2}}\\
& = & -\frac{x}{(1+2x)(1-4x^2)^{3/2}},
\een
so the second identity is prove.

For example,
\ben
c_1 & = &\frac{1}{2} a_+ +\frac{1}{2} a_- = b_+ - b_-, \\
c_2 & = &\frac{3}{8}a_+^2+\frac{1}{4}a_+a_-+\frac{3}{8}a_-^2
= 2t + 2\sqrt{t} (b_++b_-) + \biggl( \frac{3}{2} b_+^2-b_+b_-+\frac{3}{2}b_-^2 \biggr),\\
c_3 & = & \frac{5}{16}a_+^3+\frac{3}{16}a_+^2a_-+\frac{3}{16}a_-a_+^2+\frac{5}{16}a_-^3 \\
& = & 6(b_+-b_-)t+6(b_+^2-b_-^2)\sqrt{t}
+ \biggl( \frac{5}{2}b_+^3- \frac{3}{2} b_+^2b_-+\frac{3}{2} b_+b_-^2-\frac{5}{2}a_-^3
\biggr), \\
c_4 & = & \frac{35}{128}a_+^4+\frac{5}{32} a_+^3a_-
+\frac{9}{64}a_+^2a_-^2+\frac{5}{32}a_+a_-^3+\frac{35}{128}a_-^4 \\
& = & 6t^2+12(b_++b_-)t^{3/2}
+(21b_+^2-6b_+b_-+21b_-^2)t \\
& + & (15b_+^3-3b_+^2b_--3b_+b_-^2+15b_-^3)t^{1/2} \\
& + & \biggl(\frac{35}{8}b_+^4-\frac{5}{2}b_+^3b_-+\frac{9}{4}b_+^2b_-^2
-\frac{5}{2}b_+b_-^3+\frac{35}{8}b_-^4\biggr).
\een
Therefore,
after rewriting \eqref{eqn:H-1} and \eqref{eqn:H-2} in
the following form:
\bea
&&H_{-1} = - \sum_{m=0}^\infty \tilde{g}_{2m+1} c_{2m}
- \sum_{m=0}^\infty \tilde{g}_{2m+2} c_{2m+1} = 0,\\
&&H_{-2} = - \sum_{m=1}^\infty \tilde{g}_{2m} c_{2m}
- \sum_{m=0}^\infty \tilde{g}_{2m+1} c_{2m+1} = 2t,
\eea
we get:
\ben
b_+-b_-
&= & \sum_{m=0}^\infty g_{2m+1}
\biggl(\binom{2m}{m} t^m + 2m\binom{2m}{m}(b_++ b_-) t^{(2m-1)/2}+ \cdots\biggr) \\
& + & \sum_{m=0}^\infty g_{2m+2}
\biggl((2m+1)\binom{2m+1}{m}(b_+-b_-) t^{m} + \cdots \biggr),\\
2\sqrt{t} (b_++b_-) & = & -\biggl( \frac{3}{2} b_+^2-b_+b_-+\frac{3}{2}b_-^2 \biggr) \\
& + & \sum_{m=1}^\infty g_{2m}
\biggl(\binom{2m}{m} t^m + 2m\binom{2m}{m}(b_++ b_-) t^{(2m-1)/2}+ \cdots\biggr) \\
& + & \sum_{m=0}^\infty g_{2m+1}
\biggl((2m+1)\binom{2m+1}{m}(b_+-b_-) t^{m} + \cdots \biggr).
\een
This system of equations can be recursive solved as follows.
Write
\be
b_\pm = \sum_{n=1}^\infty b_\pm^{(n)},
\ee
where each $b_\pm^{(n)}$ is homogeneous of degree $n$ in $g_1, g_2, \dots$.
For example, the degree $1$ part of the system is
\ben
&& b_+^{(1)} - b_-^{(1)} = \sum_{m=0}^\infty g_{2m+1} \binom{2m}{m} t^m, \\
&& 2\sqrt{t} (b_+^{(1)}+b_-^{(1)}) =  \sum_{m=1}^\infty g_{2m} \binom{2m}{m} t^m.
\een
And so
\ben
&& b_+^{(1)} = \half\sum_{m=0}^\infty g_{2m+1} \binom{2m}{m} t^m
+ \frac{1}{4} \sum_{m=1}^\infty g_{2m} \binom{2m}{m} t^{m-1/2}, \\
&& b_-^{(1)} = -\half\sum_{m=0}^\infty g_{2m+1} \binom{2m}{m} t^m
+ \frac{1}{4} \sum_{m=1}^\infty g_{2m} \binom{2m}{m} t^{m-1/2}.
\een
Next we get the system
\ben
b_+^{(2)}-b_-^{(2)}
&= & \sum_{m=0}^\infty g_{2m+1}\cdot 2m\binom{2m}{m}(b_+^{(1)}+ b_-^{(1)}) t^{(2m-1)/2}  \\
& + & \sum_{m=0}^\infty g_{2m+2}
\cdot (2m+1)\binom{2m+1}{m}(b_+^{(1)}-b_-^{(1)}) t^{m},\\
2\sqrt{t} (b_+^{(2)}+b_-^{(2)}) & = &
-\biggl( \frac{3}{2} (b_+^{(1)})^2-b_+^{(1)}b_-^{(1)}+\frac{3}{2}(b_-^{(1)})^2 \biggr) \\
& + & \sum_{m=1}^\infty g_{2m}\cdot 2m\binom{2m}{m}(b_+^{(1)}+ b_-^{(1)}) t^{(2m-1)/2}  \\
& + & \sum_{m=0}^\infty g_{2m+1}
\cdot (2m+1)\binom{2m+1}{m}(b_+^{(1)}-b_-^{(1)}) t^{m}  .
\een
From this we can get $b_\pm^{(2)}$, and so on.
So the proof of Theorem \ref{thm:Main} is complete.

\subsection{The case of even potential function}
For even potential function
$$S(x)=-\half x^2 +\sum_{n \geq 1} g_{2n} x^{2n},$$
we have $a_+ = - a_-$,
\ben
H(z) & = &
= -\frac{S'(z)}{\sqrt{z^2-a^2_+(N)}}
= -\frac{S'(z)}{z} (1- \frac{a^2_+}{z^2})^{-1/2}  \\
& = & (1 - \sum_{m\geq 1} g_{2m}z^{2m-2})\cdot
\sum_{n=0}^\infty
\frac{1}{2^{2n}} \binom{2n}{n} \frac{a_+^{2n}(N)}{z^{2n}},
\een
from this we get
\ben
&&H_{-1} = 0,\\
&&H_{-2} = - \sum_{n=1}^\infty \tilde{g}_{2n}
\frac{1}{2^{2n}}  \binom{2n}{n} a_+^{2n}= 2t.
\een
From these we get:
\be
\sqrt{t}b_+ = - \half b_+^2
+ \half \sum_{n=1}^\infty g_n \binom{2n}{n} (t^n +2nt^{n-1/2}b_+ + \cdots),
\ee
and hence one can solve $b_+$ as a formal power series in $g_k$'s.

In the next few Subsections we will work out some concrete examples to illustrate
the applications of Theorem \ref{thm:Main} and the even potential function case.

\subsection{Fat special deformation along the $g_1$-line again}
\label{sec:Fat-g1-again}
For the action
\ben
S(z) = -\frac{1}{2}z^2 + g_1 z,
\een
one has
\ben
S'(z) = -(z - g_1).
\een
From
\ben
&& \frac{z-g_1}{\sqrt{(z-a_+)(z-a_-)}} \\
& = & (1- \frac{g_1}{z}) \cdot
\biggl(1+\frac{a_++a_- }{2z}+\frac{3a_+^2+2a_+a_-+3a_-^2}{8z^2} +\cdots \biggr)  \\
& = & 1+\frac{a_++a_--2g_1}{2z}+\frac{3a_+^2+2a_+a_-+3a_-^2-4g_1(a_++a_-)}{8z^2}
+\cdots,
\een
one gets two equations
\begin{align*}
a_++a_--2g_1 &= 0, &
3a_+^2+2a_+a_-+3a_-^2-4g_1(a_++a_-) = 16,
\end{align*}
with solutions
\begin{align}
a_++a_- & = 2g_1, & a_+a_- & = g_1^2-4t.
\end{align}
It follows that the spectral curve in resolvent is
\be
\omega = \frac{z-g_1-\sqrt{z^2-2g_1z+ g_1^2-4t}}{2}.
\ee
In this case one has
\ben
Q(z) = 1.
\een
We have recovered \eqref{eqn:Motzkin} in \S \ref{sec:Fat-g1}.

\subsection{Fat special deformation along the $g_2$-line again}
\label{sec:Fat-g2-again}
For the action function
\ben
S(z) = -\frac{1}{2}z^2 + \frac{g_2}{2} z^2
\een
one has
\be
S'(z) = -(z - g_2z).
\ee
From
\ben
&& \frac{z-g_2z}{\sqrt{(z-a_+)(z-a_-)}} \\
& = & (1- g_2) \cdot
\biggl(1+\frac{a_++a_- }{2z}+\frac{3a_+^2+2a_+a_-+3a_-^2}{8z^2}+\cdots \biggr)\\
\een
we get
\begin{align*}
a_++a_- & = 0, & a_+a_- & = \frac{-4t}{1-g_2}.
\end{align*}
The spectral curve deforms to
\ben
\omega = \frac{(1-g_2)z-(1-g_2) \sqrt{z^2-\frac{4t}{1-g_2}}}{2}
\een
with
\be
Q(z) = 1 -  g_2 .
\ee
Alternatively,
one can proceed from
\ben
&& \frac{z-g_2z}{\sqrt{z^2-a_+^2}} = (1-g_2)\cdot (1- \frac{a_+^2}{2z^2} +\cdots)
\een
to get $a_+^2 = \frac{4t}{1-g_2}$.
We have recovered \eqref{eqn:Motzkin-g2} in \S \ref{sec:Fat-g2}.

\subsection{Fat special deformation on the $(g_1,g_2)$-plane again}
\label{sec:Fat-g1g2-again}
For the action function
\be
S(z) = -\frac{1}{2}z^2 + g_1z+ \frac{g_2}{2} z^2,
\ee
one  has
\ben
&& -\frac{S'(z)}{\sqrt{(z-a_+)(z-a_-)}}
= \frac{-g_1+(1-g_2)z}{\sqrt{(z-a_+)(z-a_-)}} \\
& = & ((1- g_2) - \frac{g_1}{z}) \cdot
\biggl(1+\frac{a_++a_- }{2z}+\frac{3a_+^2+2a_+a_-+3a_-^2}{8z^2} +\cdots \biggr) \\
& = & (1-g_2) +\frac{(1-g_2)(a_++a_-)- 2g_1}{2z} \\
& + & \frac{(1-g_2)(3a_+^2+2a_+a_-+3a_-^2)-4g_1(a_++a_-)}{8z^2} +\cdots,
\een
and one gets the following two equations
\begin{align*}
(1-g_2)(a_++a_-)- 2g_1& = 0, &
(1-g_2)(3a_+^2+2a_+a_-+3a_-^2)-4g_1(a_++a_-) & = 0,
\end{align*}
with solutions:
\begin{align*}
a_++a_- & = \frac{2g_1}{1-g_2}, & a_+a_- & = \frac{-4t}{1-g_2}.
\end{align*}
The spectral curve deforms to
\be
\omega = \frac{(1-g_2)z-g_1-(1-g_2) \sqrt{z^2-\frac{2g_1z}{1-g_2}-\frac{4t}{1-g_2}
+ \frac{g_1^2}{(1-g_2)^2} }}{2}.
\ee
This recovers \eqref{eqn:omega-g1g2} in \S \ref{sec:Fat-g1g2}.

\subsection{Fat special deformation along the $g_3$-line again}
\label{sec:Fat-g3-again}
For the action function
\ben
S(z) = -\frac{1}{2}z^2 + \frac{g_3}{3} z^3,
\een
we have
\ben
&& -\frac{S'(z)}{\sqrt{(z-a_+)(z-a_-)}} = \frac{z-g_3z^2}{\sqrt{(z-a_+)(z-a_-)}} \\
& = & (1- g_3z) \cdot
\biggl(1+\frac{a_++a_- }{2z}+\frac{3a_+^2+2a_+a_-+3a_-^2}{8z^2}\\
& + & \frac{5a_+^3+3a_+^2a_-+3a_+a_-^2+5a_-^3}{16z^3}+\cdots \biggr) \\
& = & -g_3z+ (1-\half g_3(a_++a_-))
+ \frac{4(a_++a_-)-g_3(3a_+^2+2a_+a_-+3a_-^2)}{8z}  \\
& + & \frac{2(3a_+^2+2a_+a_-+3a_-^2)-g_3(5a_+^3+3a_+^2a_-+3a_+a_-^2+5a_-^3)}{16z^2}
+ \cdots,
\een
and so the following two equations:
\bea
&& 4(a_++a_-)-g_3(3a_+^2+2a_+a_-+3a_-^2) = 0, \label{eqn:g3case1} \\
&& 2(3a_+^2+2a_+a_-+3a_-^2)-g_3(5a_+^3+3a_+^2a_-+3a_+a_-^2+5a_-^3) = 32t.
\eea
From the first equation we
\be \label{eqn:g3-1}
g_3^2a_+a_- = \frac{3}{4}(g_3(a_++a_-))^2- g_3(a_++a_-),
\ee
and plug this into the second equation we find:
\be \label{eqn:g3-2}
2(g_3(a_++a_-)/2)^3 - 3 (g_3(a_++a_-)/2)^2 +(g_3(a_++a_-)/2) = 2g_3^2t.
\ee
By Lagrange inversion,
\be \label{eqn:a+a}
g_3(a_++a_-)/2
= \sum_{n=0}^\infty \frac{2^{2n}}{n+1} \binom{3n/2}{n} (2g_3^2t)^{n+1}.
\ee
The coefficients $1, 3,16,105, 768,6006,49152, \dots$ are
sequence A085614  of the On-Line Encyclopedia of Integer Sequences \cite{Sloane}.
They are the numbers of elementary arches of size $n$.
For the explicit expressions in the above equality,
see \cite{Gessel}.
The coefficients of $(a_++a_-)^2$
also have combinatorial meanings,
\ben
&& (x+ 3x^2+16x^3+105x^4+ 768x^5+6006x^6+49152x^7 + \dots)^2 \\
& = & x^2+6x^3+41x^4+306x^5+2422x^6+19980x^7+169941x^8+1479786x^9
+ \cdots
\een
the coefficients $1, 6, 41, 306, 2422, 19980, 169941, 1479786, \cdots$
are the integer sequence  A143023.
By \eqref{eqn:g3case1},
\ben
&& g_3^2a_+a_- = \frac{3}{4}(g_3(a_++a_-))^2- g_3(a_++a_-) \\
& = & \frac{3}{4} \biggl(
\sum_{n=0}^\infty \frac{2^{3n+2}}{n+1} \binom{3n/2}{n} (g_3^2t)^{n+1}
\biggr)^2
- \sum_{n=0}^\infty \frac{2^{3n+2}}{n+1} \binom{3n/2}{n} (g_3^2t)^{n+1},
\een
therefore, one gets
\be \label{eqn:aa}
\begin{split}
g_3^2a_+a_-
= & - 4 g_3^2t - \sum_{n=1}^\infty
\biggl[\frac{2^{3n+2}}{n+1} \binom{3n/2}{n} \\
- & \frac{3}{4} \sum_{i+j=n-1} \frac{2^{3n+1}}{(i+1)(j+1)}
\binom{3i/2}{i} \binom{3j/2}{j} \biggr] (g_3^2t)^{n+1}.
\end{split}
\ee
The following are the first few terms of $a_++a_-$ and $a_+a_-$ respectively.
\ben
a_++a_- & = & 4g_3t+24g_3^3t^2+256g_3^5t^3+3360g_3^7t^4 +49152g_3^9t^5 \\
& +& 768768 g_3^{11} t^6 +\cdots, \\
a_+a_- & = &   -4t -12g_3^2t^2 - 112g_3^4t^3
- 1392g_3^5t^4-19776g_3^6t^5 \\
&- & 303744g_3^{10}t^6- \cdots.
\een
From these and $a_+-a_-= ( (a_++a_-)^2-4a_+a_-)^{1/2}$,
one can compute $a_+-a_-$ and get the first few terms of $a_+$ and $a_--$:
\ben
a_+ & = & 2 t^{1/2}+2g_3t+4g_3^2t^{3/2}+12g_3^3t^2+36g_3^4t^{5/2}+128g_3^5t^3\\
 & + & 440g_3^6t^{7/2}+1680g_3^7t^4 + 6188g_3^8t^{9/2}+24576g_3^9t^5+94392g_3^{10}t^{11/}+ \cdots, \\
a_- & = & - 2 t^{1/2}+2g_3t-4g_3^2t^{3/2}+12g_3^3t^2-36g_3^4t^{5/2}+128g_3^5t^3\\
 & - & 440g_3^6t^{7/2}+1680g_3^7t^4-6188g_3^8t^{9/2}+24576g_3^9t^5-94392g_3^{10}t^{11/} + \cdots,
\een

The fat spectral curve is deformed to
\be \label{eqn:omega-g3-2}
\omega = \frac{1}{2} \biggl((z-g_3z^2) - (1-g_3z - \half g_3(a_++a_-))
\sqrt{z^2-(a_++a_-) z +a_+a_- } \biggr),
\ee
where closed formula for $a_++a_-$ and $a_+a_-$ are given by \eqref{eqn:a+a} and
\eqref{eqn:aa} respectively.
More explicitly,
\ben
\omega
= \frac{1}{2} \biggl((z-g_3z^2)
-(1-g_3z-\half g_3(4g_3t+24g_3^3t^2+256g_3^5t^3+3360g_3^7t^4 +\cdots)) \\
\cdot \sqrt{z^2-(4g_3t+24g_3^3t^2+256g_3^5t^3+3360g_3^7t^4 +\cdots)z
-4t -12g_3^2t^2 - 112g_3^4t^3 - \cdots} \biggr).
\een

Alternatively, we can solve for $b=a_++a_-$ and $c=a_+a_-$ in the following way.
Since we know the resolvent has the following two forms:
\ben
\omega&=&\frac{1}{2} \biggl((z-g_3z^2) - (1-g_3z - \half g_3b)
\sqrt{z^2-b z +c} \biggr) \\
& = & \frac{1}{2} \biggl(z -g_3z^2 -
\sqrt{(z-g_3z^2)^2 - 4(t  - g_3tz - g_3f_1)} \biggr),
\een
one can rewrite the first line in the following form:
\ben
\omega &=& \frac{1}{2} \biggl[(z-g_3z^2) -
\bigg((z-g_3z^2)^2+(g_3b+g_3^2c-\frac{3}{4}g_3^2b^2)z^2\\
& & +(-b-2g_3c+g_3b^2+g_3^2bc-\frac{1}{4}g_3^2b^3)z
+ (c-g_3bc+\frac{1}{4}g_3^2b^2c) \bigg)^{1/2} \biggr],
\een
and compare it with the second line to get the following system of equations:
\be \label{system:g3}
\begin{split}
& g_3b+g_3^2c-\frac{3}{4}g_3^2b^2 =0, \\
& -b-2g_3c+g_3b^2+g_3^2bc-\frac{1}{4}g_3^2b^3 = 4g_3t, \\
& c-g_3bc+\frac{1}{4}g_3^2b^2c = 4g_3f_1 - 4t.
\end{split}
\ee
From the first equation we get
\ben
&& c = \frac{3}{4} b^2 - \frac{1}{g_3} b.
\een
This is just \eqref{eqn:g3-1}.
Plug this into the second equation:
\ben
&& 2(g_3b/2)^3-3(g_3b/2)^2 + (g_3b/2) = 2g_3^2t,
\een
this is just \eqref{eqn:g3-2},
and so $b$ and $c$ can be solved as above.
But now we have the third equation in the above system,
from which we get:
\be
f_1 = \frac{1}{4g_3}\biggl( 4t + c(1-\frac{1}{2} g_3b)^2 \biggr).
\ee
Now we use this to recover \eqref{eqn:f1-g3}.
Recall that $b=a_++a_-$ and $c=a_+a_-$,
so by \eqref{eqn:a+a} and \eqref{eqn:aa} we have
\ben
f_1 &= & \frac{t}{g_3}+\frac{7}{16}\frac{b^2}{g_3}
-\frac{1}{4}\frac{b}{g_3^2}-\frac{1}{4}b^3+\frac{3}{64}g_3b^4 \\
& = & \frac{t}{g_3}+\frac{11}{32}\frac{b^2}{g_3}-\frac{1}{4}\frac{b}{g_3^2}
-\frac{7}{64}b^3+\frac{3}{8}bt \\
& = & \frac{1}{8}\frac{t}{g}
-\frac{1}{32}\frac{b}{g_3^2}+\frac{3}{8}bt +\frac{1}{64} \frac{b^2}{g_3} \\
& = & \frac{1}{8}\frac{t}{g} -\frac{1}{32g_3^3}
\sum_{n=0}^\infty \frac{2^{2n+1}}{n+1} \binom{3n/2}{n} (2g_3^2t)^{n+1}
+ \frac{3t}{8g_3} \sum_{n=0}^\infty \frac{2^{2n+1}}{n+1} \binom{3n/2}{n} (2g_3^2t)^{n+1} \\
& + & \frac{1}{64g_3^3} \biggl( \sum_{n=0}^\infty \frac{2^{2n+1}}{n+1}
\binom{3n/2}{n} (2g_3^2t)^{n+1} \biggr)^2 \\
& = &\frac{3}{8} \sum_{n=0}^\infty \frac{2^{3n+2}}{n+1} \binom{3n/2}{n} t^{n+2}g_3^{2n+1}
- \frac{1}{32} \sum_{n=0}^\infty \frac{2^{3n+5}}{n+2} \binom{(3n+3)/2}{n+1} t^{n+2}g_3^{2n+1}\\
& + &  \frac{1}{64} \sum_{n=0}^\infty \sum_{i+j=n}
\frac{2^{3i+2}}{i+1} \binom{3i/2}{i} \frac{2^{3j+2}}{j+1} \binom{3j/2}{j} \cdot t^{n+2}
g_3^{2n+1}.
\een
By comparing with the formula for $f_1$ in \S \ref{sec:Fat-g3},
we get a combinatorial identity:
\ben
&& \frac{3}{8} \sum_{n=0}^\infty \frac{2^{3n+2}}{n+1} \binom{3n/2}{n} x^n
- \frac{1}{32} \sum_{n=0}^\infty \frac{2^{3n+5}}{n+2} \binom{(3n+3)/2}{n+1} x^n \\
& + &  \frac{1}{64} \sum_{n=0}^\infty \sum_{i+j=n}
\frac{2^{3i+2}}{i+1} \binom{3i/2}{i} \frac{2^{3j+2}}{j+1} \binom{3j/2}{j} \cdot x^n
= \sum_{n \geq 0} 2^{2n+1}\cdot \frac{(3n)!!}{(n+2)!n!!} x^n.
\een
It does not seem to be easy to directly prove this.

\subsection{Fat special deformation on the $(g_1,g_2,g_3)$-space again}
\label{sec:Fat-g1g2g3-again}
For the action function
\ben
S(z) = -\frac{1}{2}z^2 + g_1z +\frac{g_2}{2}z^2+\frac{g_3}{3} z^3,
\een
we have
\ben
&& -\frac{S'(z)}{\sqrt{z^2-bz+c}} = \frac{z-g_1-g_2z-g_3z^2}{\sqrt{z^2-bz+c}} \\
& = & (- g_3z+ (1-g_2)-\frac{g_1}{z}) \cdot
\biggl(1+\frac{b}{2z}+ \cdots \biggr) \\
& = & -g_3z+ (1-g_2-\half g_3b)  + \cdots ,
\een
and so we get $Q(z)=1-g_2 -g_3z- \half g_3 b$
and  the fat spectral curve is deformed to
\be \label{eqn:omega-g3-3}
\omega = \frac{1}{2} \biggl((z-g_3z^2) - (1-g_2-g_3z - \half g_3b)
\sqrt{z^2-b z + c } \biggr).
\ee
By  comparing it with \eqref{eqn:omega-g1g2g3}
we get the following system of equations:
\ben
&& g_3b+g_3^2c-\frac{3}{4}g_3^2b^2 =2g_1g_3, \\
&& -(1-g_2)^2b-2(1-g_2)g_3c
+(1-g_2)g_3b^2+g_3^2bc-\frac{1}{4}g_3^2b^3
+2g_1(1-g_2)= 4g_3t, \\
&& (1-g_2)^2c-(1-g_2)g_3bc+\frac{1}{4}g_3^2b^2c
+g_1^2 = 4g_3f_1 - 4t.
\een
This system deforms the system \eqref{system:g3},
one can follow the same procedures to solve for $b$, $c$ and $f_1$.

\subsection{Fat special deformation along the $g_4$-line again}
\label{sec:Fat-g4-again}
For the action function
\ben
S(z) = - \frac{1}{2}z^2 + \frac{g_4}{4} z^3,
\een
one has:
\ben
- \frac{S'(z)}{\sqrt{z^2-a^2}}
& = & \frac{z-g_4z^3}{\sqrt{z^2-a^2}} \\
& = & \biggl((1-\frac{g_4a^2}{2})-g_4z^2 \biggr)
+ \biggl(\frac{a^2}{2}-\frac{3}{8}g_4a^4\biggr) \frac{1}{z^2} + \cdots,
\een
hence one has
\ben
\frac{a^2}{2}-\frac{3}{8}g_4a^4 = 2t.
\een
From this one gets
\ben
a^2 = \frac{2(1-\sqrt{1-12g_4t})}{3g_4}.
\een
and
\ben
Q(z) = \frac{2}{3} + \frac{1}{3} \sqrt{1-12g_4t} -g_4 z^2,
\een
and so the resolvent is given by:
\be
\omega = \frac{1}{2} \biggl(z-g_4z^3
- \biggl(\frac{2}{3} + \frac{1}{3} \sqrt{1-12g_4t} - g_4z^2\biggr)
\sqrt{z^2- \frac{2(1-\sqrt{1-12g_4t})}{3g_4}}
\biggr).
\ee
By comparing this with \eqref{eqn:omega-g4} one gets:
\ben
f_2 & = & \frac{1}{54g_4^2} \biggl((1-12g_4t)^{3/2}
- 1 + 18 g_4t \biggr)
= \sum_{n=0}^\infty 2 \cdot 3^n \frac{(2n)!}{n!(n+2)!}
g_4^n t^{n+2},
\een
the coefficients are the sequence A000168 in \cite{Sloane}.
This recovers the  formula for $f_2$ in \S \ref{sec:Fat-g4}.

\subsection{Fat special deformation along the $g_6$-line again}
\label{sec:Fat-g6-again}
For the action function
\ben
S(z) = - \frac{1}{2}z^2 + \frac{g_6}{6} z^3,
\een
one has:
\ben
- \frac{S'(z)}{\sqrt{z^2-a^2}}
& = & \frac{z-g_6z^5}{\sqrt{z^2-a^2}} \\
& = & \biggl((1-\frac{3}{8}g_6a^4)-\frac{g_6a^2}{2}z^2-g_6z^4 \biggr)
+ \biggl(\frac{a^2}{2}-\frac{5}{16}g_6a^6\biggr) \frac{1}{z^2} + \cdots,
\een
hence one has
\be \label{eqn:g6-norm}
\frac{a^2}{2}-\frac{5}{16}g_6a^6 = 2t
\ee
and
\ben
Q(z) = (1-\frac{3}{8}g_6a^4)-\frac{g_6a^2}{2}z^2-g_6z^4,
\een
and so the resolvent is given by:
\ben
\omega & = & \frac{1}{2} \biggl(z-g_6z^5
- \biggl((1-\frac{3}{8}g_6a^4)-\frac{g_6a^2}{2}z^2-g_6z^4\biggr)
\sqrt{z^2- a^2} \biggr) \\
& = & \frac{1}{2} \biggl(z-g_6z^5
- \biggl((z-g_6z^5)^2 +\big(-\frac{5}{8}a^6g_6^2+g_6a^2\big)z^4 \\
&&+\big(-\frac{15}{64}g_6^2a^8+\frac{1}{4}g_6a^4\big) z^2
+ \big(-\frac{9}{64}a^{10}g_6^2+\frac{3}{4}g_6a^6-a^2 \big)\biggr)^{1/2} \biggr) \\
\een
By comparing this with \eqref{eqn:omega-g6} one gets:
\be
\begin{split}
 \omega(z)  =  \frac{1}{2} \biggl(z-g_6z^5
 - \sqrt{ (z-g_6z^5)^2- 4(t  - g_6tz^4 - g_6 f_2z^2  - g_6f_4)} \biggr),
\end{split}
\ee
one gets a system of three equations:
\ben
&&-\frac{5}{8}a^6g_6^2+g_6a^2 = 4g_6t, \\
&& -\frac{15}{64}g_6^2a^8+\frac{1}{4}g_6a^4 = 4g_6f_2, \\
&& -\frac{9}{64}a^{10}g_6^2+\frac{3}{4}g_6a^6-a^2 = 4g_6f_4 - 4t.
\een
The first equation is just \eqref{eqn:g6-norm},
from which we get by Lagrange inversion:
\be
a^2 = \sum_{n \geq 0} \frac{1}{2n+1} \binom{3n}{n} (\frac{5}{8}g_6)^{n} (4t)^{2n+1},
\ee
the coefficients $\frac{1}{2n+1} \binom{3n}{n}$ are the sequence
A001764 in \cite{Sloane},
they have various combinatorial origins,
e.g., they are the numbers of complete ternary trees with $n$ internal nodes, or $3n$ edges.
From the second equation in the system we get
\ben
f_2 & = & -\frac{15}{256} g_6 a^8 +\frac{1}{16} a^4 \\
& = & \frac{15}{256} g_6 a^2 \cdot \frac{8}{5g_6} (4t-a^2) + \frac{1}{16} a^4 \\
& = &  \frac{3}{8} a^2t-\frac{1}{32}a^4,
\een
therefore we get a closed formula for $f_2$:
\begin{equation*}
\begin{split}
f_2 = & \frac{3}{8}t\sum_{n \geq 0}
\frac{1}{2n+1} \binom{3n}{n} (\frac{5}{8}g_6)^{n} (4t)^{2n+1}\\
- & \frac{1}{32} \biggl[ \sum_{n \geq 0}
\frac{1}{2n+1} \binom{3n}{n} (\frac{5}{8}g_6)^{n} (4t)^{2n+1}\biggr]^2.
\end{split}
\end{equation*}
After using a combinatorial identity
\ben
\frac{3}{2} \sum_{n \geq 0} \frac{1}{2n+1} \binom{3n}{n} x^n
- \frac{1}{2} \biggl[ \sum_{n \geq 0}
\frac{1}{2n+1} \binom{3n}{n} x^n \biggr]^2
= \sum_{n\geq 0} \frac{1}{(n+1)(2n+1)}\binom{3n}{n} x^n
\een
we get:
\be
f_2 = \sum_{n\geq 0} \frac{2^n}{(n+1)(2n+1)}\binom{3n}{n} (5g_6)^nt^{2n+2}
\ee
The coefficients $\frac{2^n}{(n+1)(2n+1)}\binom{3n}{n}=
1$, $1$, $4$, $24$, $176$, $1456$, $13056$, $\cdots$ are
the sequence A000309 on \cite{Sloane}
with one of the combinatorial meaning being the number of
rooted planar bridgeless cubic maps with $2n$ nodes.
By comparing with \eqref{eqn:f2-g6},
\be
\frac{1}{n!}
\corr{p_2 \big(\frac{p_6}{6}\big)^n}_0^c(t)
=  \frac{6n}{n!}
\corr{\big(\frac{p_6}{6}\big)^n}_0^c(t)
= \frac{2^n}{(n+1)(2n+1)}\binom{3n}{n} \cdot 5^nt^{2n+2}.
\ee

Similarly,
by the third equation in the system,
\ben
&& f_4 = \frac{t}{g_6}  -\frac{9}{256}a^{10}g_6+\frac{3}{16}a^6-\frac{a^2}{4g_6} \\
& = & \frac{t}{g_6}  +\frac{9}{256}a^4g_6 \cdot \frac{8}{5g_6} (4t-a^2)-\frac{3}{16}
\cdot \frac{8}{5g_6} (4t-a^2) -\frac{a^2}{4g_6} \\
& = & -\frac{9}{160}a^6  -\frac{9}{40}a^4t
+\frac{1}{20}\frac{a^2}{g_6} -\frac{1}{5}\frac{t}{g_6}  \\
& = &  \frac{9}{160} \cdot \frac{8}{5g_6} (4t-a^2)  -\frac{9}{40}a^4t
+\frac{1}{20}\frac{a^2}{g_6} -\frac{1}{5}\frac{t}{g_6}  \\
& = & \frac{9}{40}a^4t - \frac{1}{25}\frac{a^2}{g_6} + \frac{4}{25} \frac{t}{g_6}.
\een
So we get a closed formula for $f_4$:
\ben
f_4 & = & \frac{9}{40}t \biggl(\sum_{n \geq 0} \frac{1}{2n+1} \binom{3n}{n}
(\frac{5}{8}g_6)^{n} (4t)^{2n+1} \biggr)^2 \\
& - & \frac{1}{25g_6}\sum_{n \geq 0} \frac{1}{2n+1}
\binom{3n}{n} (\frac{5}{8}g_6)^{n} (4t)^{2n+1} + \frac{4}{25} \frac{t}{g_6}.
\een
The first few terms of $f_4$ are:
\ben
f_4 & = & 2t^3+24g_6t^5 +600g_6^2t^7 + 20000g_6^3t^9 \\
& + & 780000g_6^4t^{11} + 33600000g_6^5t^{13}+ \cdots.
\een

Recall the definition of $f_4$ in this case is given by \eqref{eqn:f4-g6},
so we get
\be
\begin{split}
& \sum_{n \geq 0} \frac{g_6^n}{n!}
\corr{p_4 \big(\frac{p_6}{6}\big)^n}_0^c(t) \\
= & \frac{9}{40}t \biggl(\sum_{n \geq 0} \frac{1}{2n+1} \binom{3n}{n}
(\frac{5}{8}g_6)^{n} (4t)^{2n+1} \biggr)^2 \\
 - & \frac{1}{25g_6}\sum_{n \geq 0} \frac{1}{2n+1}
\binom{3n}{n} (\frac{5}{8}g_6)^{n} (4t)^{2n+1} + \frac{4}{25} \frac{t}{g_6}.
\end{split}
\ee
So we have solved the problem of finding $f_2$ and $f_4$ in this case from \S \ref{sec:Fat-g6}.

\subsection{Fat special deformation along the $g_5$-line}
\label{sec:Fat-g5}
For the action function
\ben
S(z) = -\frac{1}{2}z^2 + \frac{g_5}{5} z^5,
\een
we have
\ben
&& -\frac{S'(z)}{\sqrt{z^2-bz+c}}
= \frac{z-g_5z^4}{\sqrt{z^2-bz+c}} \\
& = & (1- g_5z^3) \cdot
\biggl(1+\frac{b}{2z}+\frac{3b^2-4c}{8z^2}
+ \frac{5b^3-12bc}{16z^3}
+ \frac{35b^4-120b^2c+48c^2}{128z^4}+\cdots \biggr) \\
& = & 1-g_5z^3 -\half g_5 b z^2- \frac{g_5}{8}(3b^2-4c)z
- g_5 \frac{5b^3-12bc}{16} +
 \cdots,
\een
and so the following two equations:
\be
Q(z)= 1-g_5z^3 -\half g_5 b z^2- \frac{g_5}{8}(3b^2-4c)z
- g_5 \frac{5b^3-12bc}{16}
\ee
From the equation
\ben
Q(z)^2(z^2-bz+c) = S'(z)^2 - P(z)
\een
for
\begin{align*}
S(z) & = - \half  z^2 +\frac{g_5}{5}z^5, &
P(z) & = t  - g_5z^4
\biggl(\frac{t}{z} + \frac{f_1}{z^2} + \frac{f_2}{z^3}
+ \frac{f_3}{z^4}  \biggr),
\end{align*}
we get:
\ben
&& (1-g_5z^3 -\half g_5 b z^2- \frac{g_5}{8}(3b^2-4c)z
- g_5 \frac{5b^3-12bc}{16} )^2(z^2-bz+c) \\
& = & (z-g_5z^4)^2
- 4(t-g_5tz^3-g_5f_1z^2-g_5f_2z-g_5f_3).
\een
By comparing the coefficients of $z^k$,
we get a system of six equations:
\ben
&& g_5b-\frac{3}{4}g_5^2c^2-\frac{35}{64}g_5^2b^4+\frac{15}{8}g_5^2b^2c = 0, \\
&& -g_5c+\frac{1}{4}g_5b^2-\frac{7}{32}g_5^2b^5-\frac{3}{2}g_5^2bc^2
+\frac{5}{4}g_5^2b^3c = 4g_5t, \\
&& \frac{1}{8}g_5b^3+\frac{1}{4}g_5^2c^3-\frac{35}{256}g_5^2b^6
-\frac{1}{2}g_5bc-\frac{21}{16}g_5^2b^2c^2+\frac{55}{64}g_5^2b^4c
= 4g_5f_1, \\
&& -b+\frac{5}{8}g_5b^4-\frac{25}{256}g_5^2b^7+g_5c^2
-\frac{9}{4}g_5b^2c+\frac{45}{64}g_5^2b^5c \\
&& \qquad \qquad -\frac{23}{16}g_5^2b^3c^2
+\frac{3}{4}g_5^2bc^3 = 4g_5f_2, \\
&& -\frac{5}{8} g_5b^3c+c+\frac{25}{256}g_5^2b^6c-\frac{15}{32}g_5^2b^4c^2
+\frac{9}{16}g_5^2b^2c^3+\frac{3}{2}g_5bc^2
= 4g_5f_3-4t.
\een
After rewriting the first two equations in the following form:
\ben
&& b=\frac{3}{4}g_5c^2+\frac{35}{64}g_5b^4-\frac{15}{8}g_5b^2c, \\
&& c= -4t+\frac{1}{4}b^2-\frac{7}{32}g_5b^5-\frac{3}{2}g_5bc^2
+\frac{5}{4}g_5b^3c,
\een
one can solve for $b$ and $c$ recursively from the initial values:
\begin{align*}
b & = O(t^2), & c & = -4t +O(t^2),
\end{align*}
we get:
\ben
&& b = 12g_5t^2+2592g_5^3t^5+1143072g_5^5t^8 +638254080g_5^7t^{11} , \\
&& c = -4t-252g_5^2t^4-91584g_5^4t^7-47262528g_5^6t^{10} - \cdots  , \\
\een
plug these into the next equations in the system one can compute $f_1$, $f_2$, $f_4$:
\ben
&&f_1 = 2g_5t^3 +216g_5^3t^6 +63504g_5^5t^9+ 26593920g_5^7t^{12} +\cdots, \\
&&f_2 = t^2 +36g_5^2t^5 +8640g_5^4t^8 + 3312576g_5^6t^{11}  + \cdots, \\
&&f_3 = 9g_5t^4 +1512g_5^3t^7+509328g_5^5t^{10}+ \cdots.
\een

\subsection{Fat special deformation in renormalized coupling constants}

As in the thin case,
we can work with the action function written in the renormalized coupling constants:
\ben
S(z) & = & \sum_{n \geq 1} (t_{n-1} -\delta_{n,2}) \frac{x^n}{n!} \\
& = & \sum_{k=0}^\infty  \frac{(-1)^k}{(k+1)!} (I_k+\delta_{k,1}) I_0^{k+1}
+ \sum_{n \geq 2} (I_{n-1} - \delta_{n,2}) \frac{(z-I_0)^n}{n!} \\
& = & \sum_{k=0}^\infty  \frac{(-1)^k}{(k+1)!} (I_k+\delta_{k,1}) I_0^{k+1}
+ \sum_{n \geq 2} (I_{n-1} - \delta_{n,2}) \frac{w^n}{n!},
\een
where $w=z-I_0$.
From now on we work in the $w$-coordinate.
By Theorem \ref{thm:Main},
\be \label{eqn:One-Cut-w}
(S')^2-4P = R(w)^2(w^2-4pw-4q),
\ee
where $w^2-4pw-4q=z^2-bz+c$.
To find $R$,
we use the expansion
\ben
\frac{1}{\sqrt{w^2-4pw-4q}}
& = & \frac{1}{w}\sum_{m=0}^\infty \binom{2m}{m} (\frac{p}{w}+\frac{q}{w^2})^m \\
& = & \frac{1}{w} \sum_{m=0}^\infty \binom{2m}{m} \sum_{k=0}^m \binom{m}{k}
\frac{p^{m-k} q^k}{w^{m+k}} \\
& = & \sum_{n=0}^\infty \sum_{k=0}^{[n/2]} \frac{(2n-2k)!}{k!(n-k)!(n-2k)!}
\frac{p^{n-2k}q^k}{w^{n+1}},
\een
and we have
\ben
&& - \frac{S'}{\sqrt{w^2-4pw-4q}} \\
& = & - \sum_{m \geq 1} (I_m -\delta_{m,1}) \frac{w^{m-1}}{(m-1)!}
\cdot \sum_{n=0}^\infty \sum_{k=0}^{[n/2]} \frac{(2n-2k)!}{k!(n-k)!(n-2k)!}
\frac{p^{n-2k}q^k}{w^{n+1}} \\
& = & - \sum_{m \geq 1}\sum_{n=0}^\infty \sum_{k=0}^{[n/2]} \frac{(2n-2k)!}{k!(n-k)!(n-2k)!}
p^{n-2k}q^k \frac{I_m-\delta_{m,1}}{(m-1)!} w^{m-n-2},
\een
and so
\be
R(w) = \sum_{n=0}^\infty \sum_{k=0}^{[n/2]} \frac{(2n-2k)!}{k!(n-k)!(n-2k)!}
p^{n-2k}q^k \sum_{m \geq n+2} \frac{I_m-\delta_{m,1}}{(m-1)!} w^{m-n-2}.
\ee
On the other hand, by \eqref{eqn:P},
\ben
P(z) & = &  -\sum_{n \geq 2}(g_n - \delta_{n,2}) z^{n-1} \sum_{k=1}^{n-1} \frac{f_k}{z^k} \\
& = &  -\sum_{n \geq 2}(g_n - \delta_{n,2}) \sum_{k=1}^{n-1} f_k (w+I_0)^{n-1-k} \\
& = &  -\sum_{n \geq 2}(g_n - \delta_{n,2}) \sum_{l=0}^{n-2} f_{n-1-l}
\sum_{j=0}^l \binom{l}{j}I_0^{l-j}w^j \\
& = & - \sum_{j\geq 0} w^j \sum_{l\geq j} \sum_{n \geq l+2}
(g_n - \delta_{n,2}) f_{n-1-l} \binom{l}{j},
\een
where we set $f_0=t$ as a convention.
Here $g_n$ can be changed into a formal power series in $\{I_m\}$
(\cite[Proposition 2.4]{Zhou-1D}):
\be
g_{n} = \frac{t_{n-1}}{(n-1)!} = \frac{1}{(n-1)!}
\sum_{k=0}^\infty \frac{(-1)^k I_0^k}{k!}I_{n-1+k}
\ee

Now we plug all these expression into \eqref{eqn:One-Cut-w} to get:
\ben
&& \biggl(\sum_{m \geq 1} (I_m -\delta_{m,1}) \frac{w^{m-1}}{(m-1)!} \biggr)^2 \\
& - &4\sum_{j\geq 0} w^j \sum_{l\geq j} \sum_{n \geq l+2}
\biggl(\frac{1}{(n-1)!}
\sum_{k=0}^\infty \frac{(-1)^k I_0^k}{k!}I_{n-1+k}
- \delta_{n,2}\biggr) f_{n-1-l} \binom{l}{j} \\
& = & \biggl( \sum_{n=0}^\infty \sum_{k=0}^{[n/2]} \frac{(2n-2k)!}{k!(n-k)!(n-2k)!}
p^{n-2k}q^k \sum_{m \geq n+2} \frac{I_m-\delta_{m,1}}{(m-1)!} w^{m-n-2}\biggr)^2\\
&&\cdot (w^2-4pw-4q).
\een

By expanding both sides as formal power series in $w$,
one gets a sequence of equations to solve for $p$, $q$ and to compute the $f_n$'s
in the same fashion as we have done in the $g_n$-coordinates.

To illustrate the idea,
let all $I_n$ vanish except for $I_0$ and $I_2$,
i.e.,
\ben
S = - \frac{1}{2} w^2 + I_2 \frac{w^3}{3!},
\een
and so we have
\ben
&& - \frac{S'}{\sqrt{w^2-4pw-4q}}
= \frac{w- \frac{I_2}{2} w^2}{w\sqrt{1-\frac{4p}{w} - \frac{4q}{w^2}}}\\
& = & (1-\frac{I_2}{2} w) \cdot (1 + \frac{2p}{w} +\cdots) \\
& = & - \frac{I_2}{2}w + (1- I_2p ) + \cdots.
\een
From this we get
\ben
R(w) = 1- I_2p-\frac{I_2}{2}w.
\een
On the other hand, we have
\begin{align*}
g_3 & = \frac{I_2}{2}, &g_2 & = - I_0I_2, & g_1& =  \frac{1}{2}I_0^2I_2,
\end{align*}
and all other $g_n=0$, therefore,
\ben
P & = & (1-g_2) t- g_3 tz - g_3 f_1 \\
& = & (1+I_0I_2)t-  \frac{I_2}{2} t(w+I_0) - \frac{I_2}{2} f_1 \\
& = & (1+\frac{1}{2}I_0I_2) t- \frac{I_2}{2}t w - \frac{I_2}{2}f_1.
\een
Plug all the above into \eqref{eqn:One-Cut-w}:
\ben
&& (w-\frac{I_2}{2}w^2)^2-\biggl(1- I_2p-\frac{I_2}{2}w\biggr)^2(w^2-4pw-4q) \\
& = & 4 \biggl((1+\frac{1}{2}I_0I_2) t- \frac{I_2}{2}t w - \frac{I_2}{2}f_1
\biggr).
\een
By comparing the coefficients we get a system of equations:
\ben
&& I_2^2p^2+I_2^2q-2I_2 p = 0, \\
&& 4I_2^2p^3+4I_2^2pq-8I_2p^2-4I_2q+4p = -2I_2t, \\
&& 4I_2^2p^2q-8I_2pq+4q = 4 \biggl((1+\frac{1}{2}I_0I_2) t  - \frac{I_2}{2}f_1 \biggr).
\een
The situation is exactly like in \S \ref{sec:Fat-g3-again}.

\section{Concluding Remarks}
\label{sec:Conclusion}

Duality is an idea that has been widely used
in mathematics and string theory.
In this work we have just witnessed more examples
of duality in the framework of random matrix models.
The original goal for understanding the duality between
matrix models with $N \to \infty$ with matrix models with $N =1$
is almost fulfilled by considering matrix models for finite $N$
for all $N$.
It is natural to expect that for matrix models of finite $N$,
they should have some features of the theory as $N \to \infty$,
and also some features of the theory as $N=1$.
It is fortunate that such ideas can be actually worked out.
The remaining next step will be to reexamine
the procedure of taking $N \to \infty$ and in particular
the double scaling limit from our new perspectives.

In the process of studying the $N=1/N \to \infty$-duality
for matrix models,
we have found a surprising duality between the enumeration problems
for ordinary (thin) graphs and fat graphs.
More surprisingly,
this fat/thin duality comes by with a unification
by working with Hermitian matrix models for all finite $N$.

Another remarkable duality is the discrete mathematics
of enumerations of fat or thin graphs leading to various generalizations
of the Catalan numbers is dual to the algebraic geometry of some special plane curves defined
using formal power series with integral coefficients.

It is fair to say that all these dualities are emergent.
They are possible because even though we working with only 
matrices of finite sizes,
we still have an infinite degree of freedom by working with all symmetric functions.
We would like to mention that the motivation for the author
to work on emergent geometry is the wish to answer the following question:
Is mirror symmetry an emergent phenomenon?
We believe the answer is yes.
In some sense in this paper we are working in a universal setting
by working on the space of
all possible coupling constants.
It should be possible to get special examples
that are dual to the local Gromov-Witten theory of toric Calabi-Yau
three-folds.
When this is done,
one will obtain a way to embed local mirror symmetry into the emergent
geometry of matrix models.
In this way, 
local mirror symmetry will be unified with many other examples 
in the even more framework of emergent geometry of KP hierarchy explained
in \cite{Zhou-KP-Emergent} and \cite{Zhou-KP-Emergent-II}.
We expect that mirror symmetry for compact Calabi-Yau 3-folds also  
fits in a similar picture of emergent geometry,
but more complicated integrable hierarchies are needed for this purpose. 

\vspace{.2in}
{\bf Acknowledgements}.
The author is partly supported by the NSFC grant 11661131005.
Some of the results on thin emergent geometry were obtained when the author was preparing for a talk
at Russian-Chinese Conference on Integrable
Systems and Geometry, held at Euler International Mathematical Institute,
St. Petersburg. The author thanks the organizers
and the participants for the hospitality
enjoyed at this conference.
As is clear from the text we have consulted frequently
the Online Encyclopedia of Integer Sequences.

\end{document}